\documentclass[sigplan,10pt,screen]{acmart}
\newif\ifdraft\draftfalse   

\newif\ifappendix\appendixtrue   

\acmConference[PL'18]{ACM SIGPLAN Conference on Programming Languages}{January 01--03, 2018}{New York, NY, USA}
\acmYear{2018}
\acmISBN{} 
\acmDOI{} 
\startPage{1}

\setcopyright{none}

\renewcommand\footnotetextcopyrightpermission[1]{}
\usepackage{setspace}

\usepackage{bm,centernot,mathdots,mathpartir,mathtools,nicefrac,proof}
\usepackage[makeroom]{cancel}
\usepackage[frak=euler]{mathalfa}

\usepackage{relsize,xparse}
\usepackage{graphicx,transparent,wrapfig}

\usepackage{tikz}
\usepackage{subfigure}
\usetikzlibrary{arrows.meta}
\usetikzlibrary{decorations.pathmorphing}

\usepackage{booktabs, hhline, multirow, tabulary}

\usepackage[inline]{enumitem}
\usepackage{csvsimple}
\setlist{leftmargin=14pt}
\newcommand{\inliststyle}[1]{\textnormal{\textbf{\small#1}}}
\newlist{inlist}{enumerate*}{1}
\setlist[inlist]{label={\inliststyle{(\arabic*)}}}

\usepackage{listings}
\usepackage[listings]{tcolorbox}
\usepackage{accsupp}
\newcommand*{\noaccsupp}[1]{\BeginAccSupp{ActualText={}}#1\EndAccSupp{}}

\usepackage[nameinlink]{cleveref}
\crefname{algorithm}{Algorithm}{Algorithms}
\Crefname{algorithm}{Algorithm}{Algorithms}
\crefname{conjecture}{Conjecture}{Conjectures}
\Crefname{conjecture}{Conjecture}{Conjectures}
\crefname{definition}{Definition}{Definitions}
\Crefname{definition}{Definition}{Definitions}
\crefname{equation}{Equation}{Equations}
\Crefname{equation}{Equation}{Equations}
\crefname{figure}{Figure}{Figures}
\Crefname{figure}{Figure}{Figures}
\crefname{lemma}{Lemma}{Lemmas}
\Crefname{lemma}{Lemma}{Lemmas}
\crefname{section}{\S\!}{\S\S\!}
\Crefname{section}{\S\!}{\S\S\!}
\crefname{theorem}{Theorem}{Theorems}
\Crefname{theorem}{Theorem}{Theorems}

\definecolor{dkblue}{rgb}{0,0.1,0.5}
\definecolor{dkcyan}{rgb}{0.1, 0.3, 0.3}
\definecolor{dkgreen}{rgb}{0,0.3,0}
\definecolor{dkred}{rgb}{0.6,0,0}
\definecolor{dkpurple}{rgb}{0.7,0,0.4}
\definecolor{olive}{rgb}{0.4, 0.4, 0.0}
\definecolor{orange}{rgb}{0.9,0.6,0.2}

\definecolor{lightyellow}{RGB}{255, 255, 179}
\definecolor{lightgreen}{RGB}{170, 255, 220}
\definecolor{teal}{RGB}{141,211,199}
\definecolor{darkbrown}{RGB}{121,37,0}
\definecolor{princetonorange}{RGB}{255,143,0}
\definecolor{tmlblue}{RGB}{0,58,120}       

\newtcolorbox{simplebox}
             [1][]
             {arc=2pt,
              boxrule=0.75pt,
              boxsep=0.5pt,
              colback=white,
              left=8pt,
              right=8pt,
              top=0pt,
              bottom=5pt,
              #1,
             }

\newtcolorbox{inlinebox}
             [3][]
             {arc=0.5pt,
              boxrule=0.25pt,
              boxsep=0.5pt,
              colback=white,
              left=5pt,
              right=8pt,
              top=4pt,
              bottom=3pt,
              width=#2\linewidth,
              after skip=0.25em,
              valign=center,
              halign=center,
              #1,
             }
\let\oldimplies\implies
\renewcommand\implies{\mathrel{\resizebox{1.5em}{0.5em}{$\oldimplies$}}}
\let\oldiff\iff
\renewcommand\iff{\mathrel{\resizebox{1.75em}{0.5em}{$\oldiff$}}}

\newcommand{\fancyname}[1]{\textcolor{dkblue}{\textsf{\textsc{#1}}}}

\newcommand{\CF}[1]{\text{\textnormal{\lstinline!#1!}}}

\newcommand{\VARIADIC}[6]{%
  \expandafter\newcommand\csname GobbleNext#1Arg\endcsname[2]{%
    \csname CheckNext#1Arg\endcsname{##1#4##2}%
  }%
  \expandafter\newcommand\csname CheckNext#1Arg\endcsname[1]{%
    \csname @ifnextchar\endcsname\bgroup{\csname GobbleNext#1Arg\endcsname{##1}}{#2{##1#5}#6}%
  }%
  \expandafter\newcommand\csname #1\endcsname[2]{%
    \csname CheckNext#1Arg\endcsname{#3##1#4##2}%
  }%
}

\newcommand{\SYMBOL}[1]{\ensuremath{#1}}
\newcommand{\module}{\SYMBOL{\mathbf{M}}}
\newcommand{\interface}{\SYMBOL{\mathbf{F}}}
\newcommand{\spec}{\SYMBOL{\varphi}}
\newcommand{\inv}{\SYMBOL{\mathcal{I}}}
\newcommand{\verifier}{\SYMBOL{\textnormal{\textbf{\textsf{Verify}}}}}
\newcommand{\synthesizer}{\SYMBOL{\textnormal{\textbf{\textsf{Synth}}}}}

\newcommand{\Name}{\textbf{Name}}
\newcommand{\Time}{\textbf{Time}}
\newcommand{\TST}{\textbf{TST}}
\newcommand{\TVT}{\textbf{TVT}}
\newcommand{\TSC}{\textbf{TSC}}
\newcommand{\TVC}{\textbf{TVC}}
\newcommand{\MST}{\textbf{MST}}
\newcommand{\MVT}{\textbf{MVT}}
\newcommand{\Size}{\textbf{Size}}

\newcommand{\TOOL}[1]{\textnormal{\textproc{#1}}}

\newcommand{\Deryaft}{\TOOL{Deryaft}}
\newcommand{\Eldarica}{\TOOL{Eldarica}}
\newcommand{\Hanoi}{\TOOL{Hanoi}}

\newcommand{\Leon}{\TOOL{Leon}}
\newcommand{\LoopInvGen}{\TOOL{LoopInvGen}}
\newcommand{\Myth}{\TOOL{Myth}}

\newcommand{\RepInvGen}{\TOOL{Hanoi}}

\newcommand{\LinearArbitrary}{\TOOL{LinearArbitrary}}
\newcommand{\LA}{\TOOL{LA}}

\newcommand{\ALGO}[1]{\textnormal{\textproc{#1}}}

\newcommand{\IsCondInductive}{\ALGO{CondInductive}}
\newcommand{\ClosedPositives}{\ALGO{ClosedPositives}}
\newcommand{\NoNegatives}{\ALGO{NoNegatives}}

\newcommand{\VALUE}[1]{\textnormal{\texttt{#1}}}
\newcommand{\Valid}{\VALUE{Valid}}
\newcommand{\CounterExample}[1]{\VALUE{C\kern0.0625pcEx}\,#1}
\newcommand{\Failure}{\VALUE{Failure}}
\newcommand{\Success}[1]{\VALUE{Success}\;#1}

\newcommand{\OPERATOR}[1]{\ensuremath{#1}}

\newcommand{\definedas}{\OPERATOR{\mathbin{\triangleq}}}
\newcommand{\card}[1]{\OPERATOR{\mathlarger\lvert#1\mathlarger\rvert}}
\newcommand{\rank}{\OPERATOR{\mathcal{R}}}
\newcommand{\Subst}[3]{\OPERATOR{#1[#2 \mapsto #3]}}
\newcommand{\Exists}[2]{\OPERATOR{\exists \mkern2mu #1 \mkern2mu \ldotp #2}}
\newcommand{\Forall}[2]{\OPERATOR{\forall \mkern2mu #1 \mkern2mu \ldotp #2}}
\NewDocumentCommand
  {\Sufficient}
  { O{\mkern-2mu} O{} m }
  {\OPERATOR{\textnormal{\textsf{Suf}}^{\mkern3mu{#1}}_{#2}[#3]}}
\NewDocumentCommand
  {\Constructible}
  { O{} m m }
  {\OPERATOR{\bm{\mathfrak{C}}_{\mkern1mu#1}\left[#2\mkern3mu\bm{;}#3\right]}}
\newcommand{\Satisfies}[4]{\OPERATOR{\HasType{#1}{#2} \models_{\raisebox{-0.125em}{\smaller[3]\ensuremath{#3}}} #4}}
\newcommand{\EvaluatesTo}[2]{\OPERATOR{#1 \Downarrow #2}}
\newcommand{\CollectV}[2]{\OPERATOR{\left\{\mkern-3.875mu\middle|\,#2\,\middle|\mkern-3.875mu\right\}_{#1}}}
\newcommand{\LeadsTo}[5]{\OPERATOR{{#3} \bm{:} {#4} \mkern8mu\resizebox{0.9em}{0.6em}{$\blacktriangleright$}^{#1}_{#2} \mkern8mu#5}}
\newcommand{\LeadsToT}[6]{\LeadsTo{#1}{#2}{#3}{#4}{\CounterExample{\Tuple{#5}{#6}}}}

\newcommand{\EXPR}[1]{\ensuremath{#1}}
\VARIADIC{Apply}{\EXPR}{\left(}{\mkern6mu}{\right)}{}
\VARIADIC{Tuple}{\EXPR}{\left\langle}{,}{\right\rangle}{}
\newcommand{\LambdaF}[3]{\EXPR{(\lambda \HasType{#1}{#2} \ldotp #3)}}

\newcommand{\Proj}[2]{\Apply{\pi_{#1}}{#2}}
\newcommand{\W}{\EXPR{w}}
\newcommand{\expr}{\EXPR{e}}
\newcommand{\mval}{\EXPR{v_m}}
\newcommand{\mstd}{\Tuple{\ctyp}{\mval}}
\newcommand{\Vpos}{\EXPR{V_{\mkern-2mu+}}}
\newcommand{\Vneg}{\EXPR{V_{\mkern-2mu-}}}

\newcommand{\TYPE}[1]{\ensuremath{#1}}
\newcommand{\TYPENAME}[1]{\TYPE{\textnormal{\texttt{#1}}}}
\newcommand{\HasType}[2]{\TYPE{#1 \bm{:} #2}}
\NewDocumentCommand
  {\HasTypeInCtx}
  { O{} m m }
  {\TYPE{#1 \vdash \HasType{#2}{#3}}}

\newcommand{\TBool}{\TYPENAME{bool}}
\newcommand{\TInt}{\TYPENAME{int}}

\newcommand{\TList}[1]{\TYPENAME{#1\;list}}
\VARIADIC{TTuple}{\TYPE}{\left(}{\ast\,}{\right)}{}
\VARIADIC{TSum}{\TYPE}{\left(}{+}{\right)}{}
\VARIADIC{TFun}{\TYPE}{\left(}{\to}{\right)}{}
\newcommand{\ctyp}{\TYPE{\tau_c}}
\newcommand{\mtyp}{\TYPE{\tau_m}}
\newcommand{\istd}{\Exists{\alpha}{\mtyp}}

\newcommand{\HanoiNoSRC}{\Hanoi{}\textsubscript{-SRC}}
\newcommand{\HanoiNoCLC}{\Hanoi{}\textsubscript{-CLC}}
\newcommand{\OneShot}{\TOOL{OneShot}}
\newcommand{\AndStrengthen}{\TOOL{$\wedge$Str}}

\newcommand{\NumBenchmarks}{28}
\newcommand{\NumSuccess}{22}

\usepackage{algorithm}
\usepackage[noend]{algpseudocode}

\algnewcommand\ALGOrithmicswitch{\textbf{switch}}
\algnewcommand\ALGOrithmicmatch{\textbf{match}}
\algnewcommand\ALGOrithmiccase{\textbf{case}}
\algnewcommand\ALGOrithmicwith{\textbf{with}}

\algdef{SE}[SWITCH]{Switch}{EndSwitch}[1]{\ALGOrithmicmatch\ #1\ \ALGOrithmicwith}{\ALGOrithmicend\ \ALGOrithmicswitch}%
\algdef{SE}[CASE]{Case}{EndCase}[1]{$|~$ #1 $\rightarrow$}{\ALGOrithmicend\ \ALGOrithmiccase}%
\algdef{SE}[CaseTwo]{CaseTwo}{EndCaseTwo}[2]{$|~$ #1 $\rightarrow$ #2}{\ALGOrithmicend\ \ALGOrithmiccase}%
\algtext*{EndSwitch}%
\algtext*{EndCase}%
\algtext*{EndCaseTwo}%
\algtext*{EndSecondCase}%

\newlength{\algorithmicindentlength}
\algrenewcommand\algorithmicindent{\algorithmicindentlength}

\algnewcommand{\IfThen}[2]{
  \State \algorithmicif\ #1\ \algorithmicthen\ #2}
\algnewcommand{\IfThenElse}[3]{
  \State \algorithmicif\ #1\ \algorithmicthen\ #2\ \algorithmicelse\ #3}

\algrenewcommand\alglinenumber[1]{\smaller \textcolor{darkgray}{\texttt{#1}}\hspace{0.5em}}
\algnewcommand{\LineComment}[2][1]{\State {\smaller\hspace{\dimexpr \algorithmicindentlength * #1 - 2.25em \relax}\textbf{\(\blacktriangleright\)\: #2}}}
\algrenewcommand{\algorithmiccomment}[1]{\hfill {\smaller \textbf{\(\blacktriangleright\)\: #1}}}

\algrenewcommand{\Return}[1]{\State \algorithmicreturn\ #1}

\algnewcommand\algorithmicbreak{\textbf{break}}
\algnewcommand{\Break}{\State \algorithmicbreak}

\algnewcommand\algorithmicthrow{\textbf{throw}}
\algnewcommand{\Throw}[1]{\State \algorithmicthrow\ #1}

\algnewcommand\algorithmicassume{\textbf{assume}}
\algnewcommand\Assume[1]{\State\algorithmicassume\ #1}
\algnewcommand\algorithmicassert{\textbf{assert}}
\algnewcommand\Assert[1]{\State\algorithmicassert\ #1}

\algblock{Func}{EndFunc}
\algnewcommand\algorithmicfunc{\textbf{func}}
\algnewcommand\algorithmicendfunc{\textbf{end\ func}}
\algrenewtext{Func}[2]{\algorithmicfunc\ \textsc{#1}\,(#2)}
\algrenewtext{EndFunc}{\algorithmicendfunc}
\algtext*{EndFunc}

\algblock{ParFor}{EndParFor}
\algnewcommand\algorithmicparfor{\textbf{parallel\,for}}
\algnewcommand\algorithmicpardo{\textbf{do}}
\algnewcommand\algorithmicendparfor{\textbf{end\ parallel\,for}}
\algrenewtext{ParFor}[1]{\algorithmicparfor\ #1\ \algorithmicpardo}
\algrenewtext{EndParFor}{\algorithmicendparfor}
\algtext*{EndParFor}

\algblock{ForEach}{EndForEach}
\algnewcommand\algorithmicforeach{\textbf{for\,each}}
\algnewcommand\algorithmicforeachdo{\textbf{do}}
\algnewcommand\algorithmicendforeach{\textbf{end\ for\,each}}
\algrenewtext{ForEach}[1]{\algorithmicforeach\ #1\ \algorithmicforeachdo}
\algrenewtext{EndForEach}{\algorithmicendforeach}
\algtext*{EndForEach}
\colorlet{listing-comment}{gray}
\colorlet{operator-symbol}{darkbrown}

\lstdefinelanguage{OCaml}{
    language=Caml,
    morekeywords={include, module, sig, struct, val},
    morekeywords=[2]{false, true},
    keywordstyle=[2]\color{dkgreen},
    morekeywords=[3]{int, bool, list},
    keywordstyle=[3]\color{dkcyan},
    literate=%
      {=}{{{\color{operator-symbol}=}}}1
      {<}{{{\color{operator-symbol}<}}}1
      {>}{{{\color{operator-symbol}>}}}1
      {:}{{{\color{operator-symbol}:}}}1
      {;}{{{\color{operator-symbol};}}}1
      {|}{{{\color{operator-symbol}|}}}1
      {[}{{{\color{operator-symbol}[}}}1
      {]}{{{\color{operator-symbol}]}}}1
      {\&}{{{\color{operator-symbol}\&}}}1
      {->}{{{\color{operator-symbol}->}}}1
}

\lstdefinestyle{default}{
    basicstyle=\linespread{0.9}\ttfamily,
    columns=fullflexible,
    commentstyle=\sffamily\color{black!50!white},
    escapechar=\#,
    framexleftmargin=1ex,
    framexrightmargin=1ex,
    keepspaces=true,
    keywordstyle=\color{dkblue},
    mathescape,
    numbers=left,
    numberblanklines=false,
    numbersep=1.25em,
    numberstyle=\relscale{0.8}\color{gray}\ttfamily\noaccsupp,
    showstringspaces=true,
    stepnumber=1,
    xleftmargin=2.5em,
}

\lstnewenvironment{lstlistingsmall}[1][]
  {\small\lstset{#1}}
  {}

\newtcblisting
    {boxed-listing}
    [2][]
    {listing only,
     arc=2pt,
     boxrule=0.75pt,
     boxsep=0.5pt,
     colback=white,
     left=-4pt,
     right=4pt,
     top=1pt,
     bottom=0pt,
     listing options={
         xleftmargin=5.5em,
         #2,
     },
     #1,
    }

\makeatletter
\newcommand{\customlabel}[2]{%
   \protected@write \@auxout {}{\string \newlabel {#1}{{#2}{\thepage}{#2}{#1}{}} }%
   \hypertarget{#1}{#2}
}
\makeatother

\begin{document}

\title[Data-Driven Inference of Representation Invariants]
      {Data-Driven Inference of \\ Representation Invariants}

\begin{abstract}
  A \emph{representation invariant} is a property that holds of all values of
  abstract type produced by a module. Representation invariants play important
  roles in software engineering and program verification. In this paper, we
  develop a counterexample-driven algorithm for inferring a representation
  invariant that is sufficient to imply a desired specification for a module.
  The key novelty is a type-directed notion of \emph{visible inductiveness},
  which ensures that the algorithm makes progress toward its goal as it
  alternates between weakening and strengthening candidate invariants. The
  algorithm is parameterized by an example-based synthesis engine and a
  verifier, and we prove that it is sound and complete for first-order modules
  over finite types, assuming that the synthesizer and verifier are as well. We
  implement these ideas in a tool called \RepInvGen, which synthesizes
  representation invariants for recursive data types. \RepInvGen\ not only
  handles invariants for first-order code, but higher-order code as well. In its
  back end, \RepInvGen\ uses an enumerative synthesizer called \Myth\ and an enumerative testing tool as a verifier. Because \RepInvGen\ uses testing for
  verification, it is not sound, though our empirical evaluation shows that it is
  successful on the benchmarks we investigated.
\end{abstract}

\author{Anders Miltner}
\affiliation{
  \institution{Princeton University}
}
\email{amiltner@cs.princeton.edu}

\author{Saswat Padhi}
\affiliation{
  \institution{UCLA}
}
\email{padhi@cs.ucla.edu}

\author{Todd Millstein}
\affiliation{
  \institution{UCLA}
}
\email{todd@cs.ucla.edu}

\author{David Walker}
\affiliation{
  \institution{Princeton University}
}
\email{dpw@cs.princeton.edu}

%
%

%

\maketitle

\section{Introduction}%
\label{sec:introduction}

A \emph{representation invariant} is a property that holds of all
values of abstract type produced by a module.  For instance, a module
that implements a set using a list might maintain a \emph{no
  duplicates}
or \emph{is sorted} invariant over the lists.  Module implementers can rely on the invariant for correctness and efficiency and must ensure that it is maintained by each function in the module.  Making representation invariants explicit has a number of software engineering benefits:  they can be used as documentation, dynamically checked as {\em contracts}~\cite{Meyer97, Findler02}, and used for automated testing~\cite{Claessen00, Boya02}.

Representation invariants also play a key role in modular verification
of software components.  Consider a module that implements sets; its
specification \spec\ might demand that \CF{(lookup (insert s i) i)} return \CF{true} for all sets \CF{s} and items \CF{i}.  A standard way to prove such a specification~\cite{web18/appel/vfa} is in two steps:  1) prove that a predicate \inv\ is a representation invariant of the module; and 2) prove that \inv\ is {\em stronger} than \spec, i.e., all module states that satisfy \inv\ also satisfy \spec.  In other words, modular verification can be reduced to the problem of synthesizing a \emph{sufficient representation invariant}.

In this paper, we develop an approach to automatically infer a
sufficient representation invariant given a pure, functional module and a specification.   To our knowledge, the only prior work to tackle this problem~\cite{tacas07/malik/deryaft} builds candidate invariants out of a fixed set of atomic predicates and provides no correctness guarantees.  We address both of these limitations through a form of \emph{counterexample-guided inductive
synthesis} (CEGIS)~\citep[\S 5]{sttt13/solar-lezama/sketching}, which consists of an interaction between two black-box components:
\begin{inlist}
  \item a \emph{synthesizer} that generates a candidate invariant
        consistent with given sets of positive and negative examples,
  \item a \emph{verifier} that either proves that a candidate is a sufficient representation invariant
        or produces a counterexample, which becomes a new example for the synthesizer.
\end{inlist}

Our approach is inspired by recent work in data-driven inference of inductive invariants in other settings~\citep{pldi16/padhi/loopinvgen,fmsd16/sharma/c2i,
oopsla18/ezudheen/hornice,pldi18/zhu/data-driven-chc}.  As in that work, a key
challenge is how to handle {\em inductiveness counterexamples}, pairs of module states
\Tuple{s}{s'}\ such that $s$ satisfies the candidate invariant, some module operation transforms state $s$ to state $s'$, but $s'$ does not satisfy the candidate invariant.
The problem is that there are two ways to resolve such counterexamples and it is
not clear which is correct:  treat $s$ as a new negative example or treat $s'$
as a new positive example.

We observe that if $s$ is a {\em constructible} state of the module, meaning that it is reachable by a sequence of module operations, then $s'$ must be as well.  Therefore, any representation invariant will include both states, so we must treat $s'$ as a new positive example.
Based on this observation, we define a candidate invariant $I$ to be
\emph{visibly inductive} on a module relative to a set $S$ of known constructible
states if every module operation produces a
state satisfying $I$ when invoked from a state in $S$. For each candidate invariant, we first iteratively weaken
it until it is visibly inductive relative to the current set of known constructible states, in
the process adding new states to this set, and only then do we
consider other inductiveness violations. Intuitively, this approach eagerly
identifies and exploits inductiveness counterexamples for which no ``guessing''
is required.

We have formalized a general notion of inductiveness as a type-indexed logical
relation, of which both our notion of visible inductiveness and the traditional notion of (full) inductiveness are
special cases. We have also formalized our overall algorithm using this notion.
We have proven the algorithm sound and complete when the module contains
first-order code and the implementation of the abstract type is a finite domain,
provided the underlying synthesizer and verifier are also sound and complete.

We have implemented our algorithm in OCaml and call the resulting tool \RepInvGen.  To instantiate the synthesis
component of the system, we use \Myth~\cite{pldi15/osera/myth}, a type-
and example-directed synthesis engine.  \Myth\ is capable of
synthesizing invariants over recursive data types in many cases, so it
is a good fit for tackling proofs about modules that implement recursive data types,
which are the focus of our benchmarks.  To instantiate the
verification component, we use a form of enumerative test generation.
Despite the unsoundness of this
underlying verifier, our experimental results show that \RepInvGen\
still infers sufficient representation invariants in practice.  Such
{\em likely} representation invariants can be used by module
implementers and verifiers for many purposes.

We have also implemented extensions that allow \RepInvGen\ to be used
with higher-order code.  Here, the main challenge comes in how to
extract counterexamples from higher-order arguments.  It turns out that our
first-order scheme for extracting counterexamples is essentially an
application of a first-order contract that guards and logs values
passing through the first-order interface.  The solution to
counterexample-extraction from higher-order code then is to implement
higher-order contracts~\cite{Findler02} that guard and log values across this
higher-order interface.

To evaluate our tools, we constructed a benchmark suite that includes
\NumBenchmarks\ different modules, including a variety of modules over lists and
trees, many drawn from Coq libraries and books~\cite{web18/appel/vfa}.
We find that \RepInvGen\ is able to
synthesize \NumSuccess\ of these invariants within 30 minutes.

To summarize, the main contributions of this work are:
\begin{itemize}
\item An algorithm for automated synthesis of representation invariants,
  parameterized by a verifier and synthesizer.
\item A formalization of the algorithm and the key notion of visible
  inductiveness, over a first-order type theory. We prove soundness and completeness in the
  case of finite domains, if the given verifier and
  synthesizer are sound and complete.
\item An extension of the algorithm capable of extracting
  counterexamples from higher-order interfaces.
\item Implementation, optimization and evaluation of a tool called \RepInvGen\
  that synthesizes invariants over recursive data types, using an unsound,
  enumerative testing engine for verification.
\end{itemize}

\section{A Motivating Example}%
\label{sec:a-motivating-example}

In this section, we give a high-level overview of \Hanoi\ using an example.
%
Consider the interface \texttt{SET}:

\begin{lstlistingsmall}
module type SET = sig
    type t
    val empty : t
    val insert : t -> int -> t
    val delete : t -> int -> t
    val lookup : t -> int -> bool
end
\end{lstlistingsmall}

\noindent
The interface declares an abstract type \CF{t} and a number of functions
that operate over \CF{t}.
\cref{fig:set-module} shows a module \CF{ListSet} that implements the
\CF{SET} interface, using \TList{\TInt} as the concrete type.

We study the problem of verifying that \CF{ListSet}
satisfies some standard properties of sets.  An example specification follows.

{\small\arraycolsep=2pt\def\arraystretch{1.125}%
$$\begin{array}{rcl}
  \Apply{\spec}{s}
    & \definedas
    & \Forall{\HasType{i}{\TInt}} \\
    &
    & \neg \Apply{\CF{lookup}}{\CF{empty}}{i} \\
    & \wedge
    & \Apply{\CF{lookup}}{\Apply{\CF{insert}}{s}{i}}{i} \mkern7mu \wedge \mkern7mu \neg \Apply{\CF{lookup}}{\Apply{\CF{delete}}{s}{i}}{i}
\end{array}$$}

\begin{figure}[!t]\small%
\begin{boxed-listing}[left=-22pt]{}
module ListSet : SET = struct
  type t = int list#\vspace*{-0.75em}#

  let empty = []#\vspace*{-0.75em}#

  let rec lookup l x =
    match l with
    | [] -> false
    | hd :: tl -> (hd = x) || (lookup tl x)#\vspace*{-0.5em}#

  let insert l x =
    if (lookup l x) then l else (x :: l)#\vspace*{-0.5em}#

  let rec delete l x =
    match l with
    | [] -> []
    | hd :: tl -> if (hd = x) then tl
                  else (hd :: (delete tl x))
end
\end{boxed-listing}
\caption{A module that implements \texttt{SET} using lists.}%
\label{fig:set-module}
\end{figure}

\noindent
Note that this specification does \emph{not} hold for arbitrary integer lists.
For example, $\Apply{\CF{lookup}}{\Apply{\CF{delete}}{\CF{[1;1]}}{\CF{1}}}{\CF{1}}$ returns \CF{true}.
Nonetheless, the \CF{ListSet} module is a correct implementation of the \CF{SET} interface,
because the specification holds for all values of the abstract type
\CF{t}
that the module can actually construct.  Such values are usually
called the \emph{representations} of the abstract type \CF{t}.  To emphasize that such values can be constructed by execution
of module operations, we say a value $v$ is 
 \emph{constructible} at type $\tau$ whenever a client with access to the
 module can produce $v$ at the type $\tau$. 


A standard approach~\citep{web18/appel/vfa} to prove
that a module implementation satisfies such a specification is to identify a
\emph{sufficient representation invariant}.
In our example, such an invariant for \CF{ListSet}
is a predicate $\HasType{\inv_\star}{\TFun{\TList{\TInt}}{\TBool}}$ that is
\begin{itemize}
  \item sufficient for \spec,
        i.e. $\Forall{\HasType{s}{\TList{\TInt}}}{\Apply{\inv_\star}{s}
          \implies \Apply{\spec}{s}}$, and
          \item whenever operations of \CF{ListSet} module are
            supplied with argument values of abstract type that satisfy the
            invariant, they produce values of abstract type that
            satisfy the invariant, \emph{i.e.}, the module is \emph{inductive}
            with respect to the invariant.
\end{itemize}
          
In other words, $\inv_\star$ contains all integer lists that are representations of type \CF{t},
and is contained in the set of integer lists that satisfy \spec.
\cref{fig:rep-invariant} shows this relationship pictorially.

For \CF{ListSet}, the predicate demanding an integer list has no
duplicates is a sufficient representation invariant for \spec.
Our tool \RepInvGen\ automatically generates that invariant:

\begin{lstlistingsmall}
  let rec $\inv_\star$ : int list -> bool = function
    | [] -> true
    | hd :: tl -> (not (lookup tl hd)) && ($\inv_\star$ tl)
\end{lstlistingsmall}


\subsection{Overview of \RepInvGen}%
\label{sec:a-motivating-example.overview-of-repinvgen}

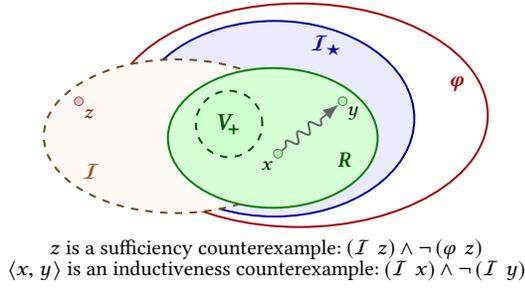
\begin{figure}[!t]
    \centering%
    {\smaller\vspace*{0.25em}%
     $\Vpos$ = a set of known constructible values \hfill \inv\ = a candidate invariant}\\[0.375em]
    \resizebox{0.7\linewidth}{!}{%
        \begin{tikzpicture}[
            dot/.style={circle,inner sep=1.25pt,fill=green!32!white,draw=gray!80!black,name=#1},
            arrow/.style={-Latex,thick,draw=gray!80!black,decoration={snake,segment length=6pt,amplitude=2pt,pre length=2pt,post length=8pt}}
        ]
            \draw[thick, draw=red!64!black]
                (0.25,-0.075) ellipse (27mm and 16mm)
                node[anchor=north east, yshift=7mm, xshift=25mm]
                {\color{red!64!black}\small $\bm{\spec}$};
            \draw[thick, fill=blue!8!white, draw=blue!64!black]
                (-0.1,-0.125) ellipse (2cm and 14mm)
                node[anchor=east, xshift=11mm, yshift=10.5mm]
                {\color{blue!64!black}\small $\bm{\inv_\star}$};
            \draw[thick, dashed, fill=orange!5!white, fill opacity=0.5, draw=orange!48!black]
                (-1.5,-0.375) ellipse (19mm and 11mm)
                node[anchor=east, xshift=-1cm, yshift=-5mm, opacity=1]
                {\color{orange!64!black}\small $\bm{\inv}$};
            \draw[thick, fill=green!16!white, fill opacity=0.25, draw=green!50!black]
                (-0.125,-0.4) ellipse (15mm and 10mm)
                node[anchor=west, xshift=8mm, yshift=-3mm, opacity=1]
                {\color{green!32!black}\small $\bm{R}$};
            \draw[thick, dashed, fill=green!20!white, fill opacity=0.25, draw=green!32!black]
                (-0.75,-0.2) circle (4.75mm)
                node[anchor=center, opacity=1]
                {\color{green!32!black}$\bm{\Vpos}$};

            \node [dot=A1] at (-0.05,-0.625) {};
            \node [below of=A1, xshift=-1.5mm, yshift=8.25mm]
                  {\color{green!8!black}\smaller$x$};
            \node [dot=A2] at (0.875,0.125) {};
            \node [below of=A2, xshift=1.5mm, yshift=8.25mm]
                  {\color{green!8!black}\smaller$y$};

            \path [draw, arrow, decorate] (A1) -- (A2);

            \node [dot=B1,fill=red!25!white] at (-2.9,0.125) {};
            \node [below of=B1, xshift=1.5mm, yshift=8.25mm]
                  {\color{red!48!black}\smaller$z$};
        \end{tikzpicture}} \\
    {\smaller%
     $z$ is a sufficiency counterexample: $\Apply{\inv}{z} \wedge \neg \Apply{\spec}{z}$ \\[-0.25em]
     \Tuple{x}{y}\ is an inductiveness counterexample: $\Apply{\inv}{x} \wedge \neg \Apply{\inv}{y}$\\[0.375em]}
    \caption{A sufficient representation invariant $\inv_\star$ implies the spec
             and is an overapproximation of the set of
             representations $R$ of the module's abstract type.}%
    \label{fig:rep-invariant}
\end{figure}

Given a module, an interface and a specification,
\RepInvGen\ employs a form of \emph{counterexample-guided inductive
synthesis} (CEGIS)~\citep[\S 5]{sttt13/solar-lezama/sketching} to infer a sufficient representation invariant.  Specifically, we use a generate-and-check approach that iterates between two black-box components:
\begin{inlist}
  \item a synthesizer \synthesizer\ that generates a candidate invariant, which is a predicate that \emph{separates} a
        set \Vpos\ of positive and a set \Vneg\ of negative examples, and
  \item a verifier \verifier\ that checks if a candidate invariant satisfies the desired properties
        and otherwise generates a counterexample.
\end{inlist}
CEGIS has been successfully applied
to other forms of invariant inference~\citep{pldi16/padhi/loopinvgen,fmsd16/sharma/c2i,
oopsla18/ezudheen/hornice,pldi18/zhu/data-driven-chc}.

We illustrate \RepInvGen\ and its key challenges via our running example.
Initially the \Vpos\ and \Vneg\ sets are empty,
so suppose that \synthesizer\ generates the candidate invariant \CF{fun _ -> true}.
This invariant is inductive, but not sufficient.
Hence \verifier\ will provide a counterexample, for instance \CF{[1;1]},
which is an integer list that satisfies the candidate invariant
but not the specification \spec\ (see $z$ in \cref{fig:rep-invariant}).  
As the final invariant must imply \spec, \CF{[1;1]} is added to \Vneg,
which forces \synthesizer\ to choose a stronger candidate invariant in the next round.


The main challenge in using this approach is the need to handle counterexamples to inductiveness.
For instance, suppose that at some point during the algorithm we have 
$\Vpos = \{ \CF{[]}, \CF{[3]} \}$ and $\Vneg = \{ \CF{[1;1]} \}$,
and suppose \synthesizer\ generates the following candidate invariant:

\begin{lstlistingsmall}
  let $\inv$ : int list -> bool = function
    | [] -> true
    | hd :: _ -> hd <> 1
\end{lstlistingsmall}

\noindent
This candidate is not inductive over the \CF{ListSet} module.
For instance, \CF{[0]} satisfies the candidate,
but $\Apply{\CF{insert}}{\CF{[0]}}{\CF{1}} = \CF{[1;0]}$ does not.
Hence the pair \Tuple{\CF{[0]}}{\CF{[1;0]}}\ constitutes an
inductiveness counterexample (see \Tuple{x}{y}\ in
\cref{fig:rep-invariant}).  Resolving such an inductiveness counterexample requires ensuring
that either both $x$ and $y$ satisfy the candidate invariant or that neither does.  This leads to two possibilities, and the problem is that it's unclear which one is correct:
\begin{itemize}
  \item add \CF{[0]} to \Vneg\ so that it will be excluded from the next candidate invariant
  \item add \CF{[1;0]} to \Vpos\ so that it will be included in the next candidate invariant
\end{itemize}

However, observe that if \Tuple{x}{y}\ is an inductiveness counterexample
and $x$ is known to be constructible, then $y$ is constructible as well.
Since a representation invariant must include all constructible integer lists,
there is no choice to make in this case:  we must add $y$ to \Vpos.  

Based on this observation, our algorithm maintains the property that all elements of \Vpos\ both
satisfy the current candidate invariant and are known to be constructible.
To verify inductiveness of a candidate invariant, we first check a property that we call \emph{visible inductiveness},
which informally requires that there are no inductiveness counterexamples \Tuple{x}{y}\ such that $x \in \Vpos$.
If such a counterexample exists, we add $y$ to \Vpos, ask \synthesizer\ for a new candidate invariant,
and re-check visible inductiveness on the updated \Vpos.

In our running example, the candidate invariant \inv\ shown above is not visibly inductive,
so \verifier\ would produce a counterexample, for instance \Tuple{\CF{[]}}{\CF{[1]}}.
Unlike the case for the counterexample \Tuple{\CF{[0]}}{\CF{[1;0]}}\ shown earlier,
by construction the first element of this new pair is in \Vpos,
so we know that we must add \CF{[1]} to \Vpos.
We then re-check visible inductiveness, continuing in this way until the candidate invariant is visibly inductive on \Vpos.

At that point, we check full inductiveness.  Because \inv\ is visibly inductive
with respect to \Vpos,
any counterexample to full inductiveness is a pair \Tuple{x}{y}\ where $x$ is not in \Vpos.  
In this case, in order to maintain the invariant that \Vpos\ only contains constructible values
we resolve the counterexample by adding $x$ to \Vneg.
So in general, the elements of \Vneg\ all falsify the current candidate invariant,
but they may or may not be constructible.
With this new negative example, \synthesizer\ will produce a stronger candidate invariant.
We then restart the process all over again,
first weakening this new invariant to be visibly inductive and then strengthening it to be inductive.

In \cref{sec:the-repinvgen-system.properties} we show that despite this interplay between
weakening and strengthening, \RepInvGen\ is sound and complete over finite domains
if \verifier\ and \synthesizer\ are sound and complete.
That is, if a sufficient representation invariant exists then \RepInvGen\ will produce one.

The question of how to handle inductiveness counterexamples arises in prior work on data-driven invariant inference.
Some of this work also observes that if $x$ is constructible in an inductiveness counterexample \Tuple{x}{y},
then so is $y$~\citep{fmsd16/sharma/c2i, pldi18/zhu/data-driven-chc}.
However, those approaches only leverage this observation opportunistically,
when a counterexample to full inductiveness happens to satisfy it.
In contrast, we define the notion of visible inductiveness
and use this notion to eagerly weaken a candidate invariant until
no such counterexamples exist.  
We demonstrate empirically in Section~\ref{sec:experiments} that our eager search for visible 
inductiveness counterexamples provides performance benefits.
We also prove
a completeness result for our approach in the context of finite domains, which those prior approaches lack.  
To our knowledge, the only prior CEGIS-based approaches to inductive
invariant inference that have a completeness result
depend upon special-purpose synthesizers
that directly accept inductiveness counterexamples in addition to
positive and negative examples~\citep{popl16/garg/ice-dt, oopsla18/ezudheen/hornice}.


\subsection{Handling Binary Functions}%
\label{sec:a-motivating-example.handling-binary-functions}

Consider the following extension to our \CF{SET} interface,
which exposes additional functions for set union and intersection:

\begin{lstlistingsmall}
  module type ESET = sig
    include SET
    val union : t -> t -> t
    val inter : t -> t -> t
  end
\end{lstlistingsmall}

Consider an extension of the \CF{ListSet} module that supports these functions
(implementation not shown in the interest of space).
When verifying inductiveness, an inductiveness counterexample on either \CF{union} or \CF{inter}
is now a \emph{triple} \Tuple{x_1}{x_2}{y}.
This increases the number of possible ways to resolve the counterexample to four:
\begin{inlist}
  \item add $x_1$ to \Vneg,
  \item add $x_2$ to \Vneg,
  \item add both $x_1$ and $x_2$ to \Vneg, or
  \item add $y$ to \Vpos.
\end{inlist}
More generally, the number of choices grows exponentially in the number of arguments
to the function that have type \CF{t}.

\RepInvGen\ naturally extends to this setting:
By construction, a counterexample to visible inductiveness due to \CF{union} or \CF{inter}
will be a triple \Tuple{x_1}{x_2}{y}\ where $x_1$ and $x_2$ are in \Vpos, so as before we add $y$ to \Vpos.  
On the other hand, a counterexample to inductiveness due to \CF{union} or \CF{inter}
will be a triple \Tuple{x_1}{x_2}{y}\ where at least one of $x_1$ and $x_2$ is \emph{not} in \Vpos.
In this case, we simply add each $x_i$ that is not in \Vpos\ to \Vneg.

\RepInvGen\ handles $n$-ary specifications in a similar manner.
For instance, we may want to prove that a module implementing the \CF{ESET} interface
satisfies the following specification:

{\small\arraycolsep=2pt\def\arraystretch{1.125}%
$$\begin{array}{rcl}
  \Apply{\spec'}{s_1}{s_2}
    & \definedas
    & \Forall{\HasType{i}{\TInt}}{} \\
    &
    & (\Apply{\CF{lookup}}{s_1}{i} \vee \Apply{\CF{lookup}}{s_2}{i} \\
    &
    & \mkern12mu \implies \Apply{\CF{lookup}}{\Apply{\CF{union}}{s_1}{s_2}}{i}) \\
    & \wedge
    & (\Apply{\CF{lookup}}{s_1}{i} \wedge \Apply{\CF{lookup}}{s_2}{i} \\
    &
    & \mkern12mu \implies \Apply{\CF{lookup}}{\Apply{\CF{inter}}{s_1}{s_2}}{i})
\end{array}$$}
If a candidate invariant is not strong enough to imply this specification, then a counterexample will consist of a pair \Tuple{x_1}{x_2}\ where at least one of $x_1$ and $x_2$ is not in \Vpos.
In this case, we again add each $x_i$ that is not in \Vpos\ to \Vneg.

Our algorithm remains sound and complete for finite domains in the presence of these extensions, assuming the verifier and synthesizer are as well.



\section{The Inference Algorithm}%
\label{sec:the-repinvgen-system}

In this section, we describe our algorithm formally
and characterize its key properties.

\subsection{Preliminaries}%
\label{sec:the-repinvgen-system.preliminaries}

Our programming language is a first-order variant of the simply-typed
lambda calculus with functions, pairs, a base type ($\beta$) and a
single designated abstract type ($\alpha$).  
The syntax of 0-order types ($\sigma$), 1st-order types ($\tau$),
values ($v$) and expressions $e$ are provided below.
$$\def\arraystretch{1.25}%
\begin{array}{rccl}
  \text{(0-types)} & \sigma & \Coloneqq & \beta \mid \alpha \mid \TTuple{\sigma}{\sigma} \\
  \text{(1-types)} & \tau   & \Coloneqq & \sigma \mid \sigma \rightarrow \tau \mid \TTuple{\tau}{\tau} \\
  \text{(values)}  & v      & \Coloneqq & \W \mid \Tuple{v_1}{v_2} \mid \LambdaF{x}{\sigma}{e} \\
  \text{(expressions)}  & e      & \Coloneqq & x \mid v \mid \Proj{i}{e} \mid \Apply{e_1}{e_2}
\end{array}$$

We use $x$ for value variables and $\W$ for constants of the base type $\beta$.
The expression $\Proj{i}{e}$ is the $i^\text{th}$ projection from the pair $e$.
We write \HasTypeInCtx[\Gamma]{\expr}{\tau} to indicate that an expression \expr\ has type $\tau$
in the context $\Gamma$, which maps variables to their types.
We write \HasTypeInCtx{\expr}{\tau} when $\Gamma$ is empty, as
will be the case in most of this work.  We write
$\Subst{\tau}{\alpha}{\ctyp}$ to substitute $\ctyp$ for $\alpha$ in $\tau$.
Finally, and we use \EvaluatesTo{\expr}{v} to indicate that \expr\ evaluates to $v$.
We refer the reader to \citet{book/pierce/tapl} for the details.

We assume a module defines a single abstract type ($\alpha$), which is
declared in its interface.
A \emph{module interface} ($\interface = \istd$) is a pair of a name
($\alpha$) for the abstract data type and a signature \mtyp\ that specifies the
types for operations over the abstract type.
A \emph{module implementation} $\module = \mstd$ is
the classic existential package of a concrete type \ctyp\
and a value \mval\ containing operations over the type \ctyp.
We say a module \mstd\ \emph{implements} an interface \istd\
when it is well-typed as per the usual rules for
existential introduction~\cite[\S 24]{book/pierce/tapl},
i.e. \HasTypeInCtx {\mval} {\Subst{\mtyp}{\alpha}{\ctyp}}.

In addition to an interface, 
we also assume the existence of a target \emph{specification} \spec,
which captures the desired correctness criteria for a module
implementation.
These specifications are universal properties of the values of the abstract type; we formalize them as polymorphic functions over the module operations, \emph{i.e.}, $\HasType {\spec} {\Forall {\alpha} {\TFun{\mtyp}{\alpha}{\TBool}}}$.
We saw an example specification for integer sets in \cref{sec:a-motivating-example}.  


The values of an abstract type $\alpha$ are simply the values that are \emph{constructible} at type $\alpha$ through the module interface.  Below we define the notion of a $\tau$-constructible value and then use it to define when a module satisfies a specification.

\begin{definition}[$\tau$-Constructible Value:  {\Constructible[\module]{v}{\tau}}]%
  \label{def:constructible}
  A value $v$ is $\tau$-constructible using \module,
  denoted \Constructible[\module]{v}{\tau},
  iff there exists a function $\HasType{f}{\Forall{\alpha}{\TFun{\mtyp}{\tau}}}$
  such that $\EvaluatesTo{\Apply{f[\tau_c]}{\mval}}{v}$.
\end{definition}


\begin{definition}[Specification Satisfaction: {\Satisfies{\module}{\interface}{}{\spec}}]%
  \label{def:module-spec-satisfaction}
  A module \module\ with interface \interface\
  is said to satisfy a given specification \spec,
  denoted \Satisfies{\module}{\interface}{}{\spec},
  iff every $\alpha$-constructible value satisfies \spec,
  i.e. $ \Forall{\HasType{v}{\ctyp}}{\Constructible[\module]{v}{\alpha}
                                     \implies
                                     \Apply{\spec[\ctyp]}{\mval}{v}} $.
                                 \end{definition}

\subsection{Representation Invariants}%
\label{sec:the-repinvgen-system.representation-invariants}

\begin{figure*}[t]\small%
\begin{simplebox}[top=0.875em]
  \begin{mathpar}
    \inferrule{}
              {\LeadsTo{P}{Q}
                       {\W}{\beta}
                       {\Valid}}
              \quad(\customlabel{rule:i-base-valid}{\fancyname{I-B}})
    \and
    \inferrule{\HasTypeInCtx{v}{\ctyp}
               \and
               \Apply{Q}{v}}
              {\LeadsTo{P}{Q}
                       {v}{\alpha}
                       {\Valid}}
              \quad(\customlabel{rule:i-abs-valid}{\fancyname{I-A}})
    \and
    \inferrule{\LeadsTo{P}{Q}{v_1}{\tau_1}{\Valid}
               \and
               \LeadsTo{P}{Q}{v_2}{\tau_2}{\Valid}}
              {\LeadsTo{P}{Q}
                       {\Tuple{v_1}{v_2}}{\TTuple{\tau_1}{\tau_2}}
                       {\Valid}}
              \quad(\customlabel{rule:i-prod-valid}{\fancyname{I-Prod}})
    \\
    \inferrule{\HasTypeInCtx{v}{\Subst{\TFun{\sigma_1}{\tau_2}}{\alpha}{\ctyp}}
               \and
               \Forall{v_1}
                      {\Forall{v_2}
                              {\left(
                                 \LeadsTo{Q}{P}{v_1}{\sigma_1}{\Valid}
                                 \enskip\wedge\enskip
                                 \EvaluatesTo{\Apply{v}{v_1}}{v_2}
                                 \implies
                                 \LeadsTo{P}{Q}{v_2}{\tau_2}{\Valid}
                               \right)}}}
              {\LeadsTo{P}{Q}
                       {v}{\TFun{\sigma_1}{\tau_2}}
                       {\Valid}}
              \quad(\customlabel{rule:i-fun-valid}{\fancyname{I-Fun}})
            \end{mathpar}%
            \begin{centering}
              \dotfill 
                \end{centering}
  \begin{mathpar}
    \inferrule{\HasTypeInCtx{v}{\ctyp}
               \and
               \neg \Apply{Q}{v}}
              {\LeadsToT{P}{Q}
                        {v}{\alpha}
                        {\{\}}{\{ v \}}}
              \quad(\customlabel{rule:i-abs-cex}{\fancyname{I-A-CEx}})
    \\
    \inferrule{\LeadsToT{P}{Q}{v_1}{\tau_1}{S}{V}
               \and
               \HasTypeInCtx{v_2}{\Subst{\tau_2}{\alpha}{\ctyp}}}
              {\LeadsToT{P}{Q}
                        {\Tuple{v_1}{v_2}}{\TTuple{\tau_1}{\tau_2}}
                        {S}{V}}
              \quad(\customlabel{rule:i-prod-cex-1}{\fancyname{I-Prod-CEx\textsubscript{\kern0.035pc1}}})
    \and
    \inferrule{\HasTypeInCtx{v_1}{\Subst{\tau_1}{\alpha}{\ctyp}}
               \and
               \LeadsToT{P}{Q}{v_2}{\tau_2}{S}{V}}
              {\LeadsToT{P}{Q}
                        {\Tuple{v_1}{v_2}}{\TTuple{\tau_1}{\tau_2}}
                        {S}{V}}
              \quad(\customlabel{rule:i-prod-cex-2}{\fancyname{I-Prod-CEx\textsubscript{\kern0.04pc2}}})
    \\
    \inferrule{\HasTypeInCtx{v}{\Subst{\TFun{\sigma_1}{\tau_2}}{\alpha}{\ctyp}}
               \and
               \LeadsTo{Q}{P}{v_1}{\sigma_1}{\Valid}
               \and
               \EvaluatesTo{\Apply{v}{v_1}}{v_2}
               \and
               \LeadsToT{P}{Q}{v_2}{\tau_2}{S}{V}}
              {\LeadsToT{P}{Q}
                        {v}{\TFun{\sigma_1}{\tau_2}}
                        {\CollectV{\sigma_1}{v_1} \cup\, S}{V}}
              \quad(\customlabel{rule:i-fun-cex}{\fancyname{I-Fun-CEx}})
            \end{mathpar}%
                        \begin{centering}
              \dotfill 
                \end{centering}
  \begin{mathpar}
    \inferrule{}
              {\CollectV{\beta}{\W} = \{\}}
              \quad(\customlabel{rule:c-base}{\fancyname{C-Base}})
    \and
    \inferrule{}
              {\CollectV{\alpha}{v} = \{v\} \quad \textrm{if}\ \HasTypeInCtx{v}{\ctyp}}
              \quad(\customlabel{rule:c-abs}{\fancyname{C-Abs}})
    \and
    \inferrule{}
              {\CollectV{\TTuple{\sigma_1}{\sigma_2}}{\Tuple{v_1}{v_2}}
               =
               \CollectV{\sigma_1}{v_1} \cup\, \CollectV{\sigma_2}{v_2}}
              \quad(\customlabel{rule:c-prod}{\fancyname{C-Prod}})
  \end{mathpar}
\end{simplebox}
\caption{Inference rules for conditional inductiveness.}%
\label{fig:rules}
\end{figure*}

Loosely speaking, a representation invariant is a property that is
preserved by operations over the abstract type of a module.
As such, we say that a representation invariant is a \emph{fully
  inductive} property of a module.  The first part of \cref{fig:rules} defines a relation that we call \emph{conditional inductiveness}, which is a generalization of both full inductiveness and the notion of visible inductiveness described earlier. Specifically, the relation $\LeadsTo{P}{Q}{v}{\tau}{\Valid}$ may be read as
``value $v$ is conditionally inductive at type $\tau$ with respect to
properties $P$ and $Q$.''

\paragraph{Full inductiveness.}
When $P$ and $Q$ are the same property $I$ (\emph{i.e.}, $P = Q = I$) ,
these rules correspond to the standard logical relation over closed
values for System F~\cite{AmalLR}, but where there is exactly one free type variable
($\alpha$) and that type variable is associated with the concrete type
$\tau_c$ and the unary relation $I$.  Values of the
abstract type $\alpha$ are in the relation if they satisfy $I$ (rule \ref{rule:i-abs-valid}).
Products satisfy the relation if their components satisfy the relation (rule \ref{rule:i-prod-valid}).
Functions satisfy the relation
if they take arguments in the relation to results in the relation (rule \ref{rule:i-fun-valid}).

The following corollary of Reynolds' theory of parametricity~\cite{ifip83/reynolds/parametricity}
says that if $I$ is a representation invariant then all $\alpha$-constructible values satisfy it.

\begin{corollary}\label{lem:repinvariant-subsumes-constructible}
  $$
    \LeadsTo{\inv}{\inv}{\mval}{\mtyp}{\Valid}
    \implies
    (\Forall{\HasType{v}{\ctyp}}
            {\Constructible[\module]{v}{\alpha} \implies \Apply{\inv}{v}})
  $$
\end{corollary}

\noindent
Therefore, to prove that a module meets a specification it is enough to identify a \emph{sufficient representation invariant}.

\begin{definition}[Sufficient Predicate: {\Sufficient[\spec][\module]{p}}]\label{def:sufficiency}
  A predicate $\HasType{p}{\TFun{\ctyp}{\TBool}}$ is
  \emph{sufficient for proving that \module\ satisfies \spec{}},
  denoted \Sufficient[\spec][\module]{p},
  iff $\Forall{\HasType{v}{\ctyp}}{\Apply{p}{v} \implies \Apply{\spec\,[\ctyp]}{\mval}{v}}$.
\end{definition}

\begin{definition}[Sufficient Representation Invariant]\label{def:sufficient-rep-inv}
  A predicate $\HasType{\inv}{\TFun{\ctyp}{\TBool}}$ is called a sufficient representation invariant
 for a module \module\ with respect to a specification \spec{},
  denoted \Satisfies{\module}{\interface}{\inv}{\spec},
 iff $\Sufficient[\spec][\module]{\inv} \ \wedge\
 \LeadsTo{\inv}{\inv}{\mval}{\mtyp}{\Valid}$.
\end{definition}

\begin{theorem}\label{thm:exists-repinv-implies-constructible}
  If a sufficient representation invariant exists,
  then the module satisfies the specification,
  i.e.
  \begin{align*}
    & \left( \Exists{\HasType{\inv}{\TFun{\ctyp}{\TBool}}}{\Satisfies{\module}{\interface}{\inv}{\spec}} \right) \\
    \implies & \left( \Forall{\HasType{v}{\ctyp}}{\Constructible[\module]{v}{\alpha} \implies \Apply{\spec\,[\ctyp]}{\mval}{v}} \right)
  \end{align*}
\end{theorem}

\begin{proof}
  Follows from Corollary~\ref{lem:repinvariant-subsumes-constructible}, \cref{def:sufficiency}, and \cref{def:sufficient-rep-inv}.
\end{proof}

\paragraph{Conditional inductiveness.}
When $P$ and $Q$ are not the same, conditional inductiveness informally requires that if the client supplies values of abstract type satisfying $P$ then
the module will produce values of abstract type satisfying $Q$.  When
conditional inductiveness is used to check visible inductiveness in our algorithm, $P$ will
be the set $V_+$ of examples that are known to be $\alpha$-constructible by the
module and $Q$ will be a candidate representation invariant.  
The most interesting rule when $P$ and $Q$ are different is the \ref{rule:i-fun-valid} rule for functions.
Specifically, notice the inversion of $P$ and $Q$ in the negative position:  If the argument is a value of abstract type, it must satisfy
$P$, not $Q$.  In other
words, this element of the formalism codifies the intuition that if the
client supplies values that satisfy $P$ then the module will supply
values that satisfy $Q$.

\paragraph{Counterexamples.}
Normally logical relations are only used to prove that an invariant is inductive.  However, we additionally require counterexamples from failed inductiveness checks, to drive our CEGIS-based invariant inference algorithm.  The second section of
\cref{fig:rules} provides the logic for refuting conditional
inductiveness and generating counterexamples.
This judgement has the form $\LeadsToT{P}{Q}{v}{\tau}{S}{V}$, which can be read as ``value $v$ is not conditionally inductive at type $\tau$ with respect to
properties $P$ and $Q$, with inductiveness counterexample witnesses $S$ and
$V$.''  Here the set $S$ contains values that satisfy $P$, the set $V$
contains values that falsify $Q$, and intuitively the values in $V$
can be computed using module operations given inputs from $S$.

As an example, consider the rule for values of abstract type
(rule \ref{rule:i-abs-cex}).
Here, a value $v$ of type $\alpha$ is not conditionally inductive if it
falsifies $Q$.  The counterexample produced includes $v$ in the set
$V$ (and returns the empty $S$), and hence satisfies the judgemental invariant explained above.
As another example, a function is not conditionally inductive (rule \ref{rule:i-fun-cex}) if there is an argument $v_1$ in the relation that causes the function to produce a result $v_2$ that is not in the relation.  In that case, the function 
$\CollectV{\sigma}{v}$ is used to collect all values of type $\alpha$ in $v_1$ to put in the returned set $S$, since they are the inputs that led to the result $v_2$.

The completeness of our algorithm for inferring sufficient representation invariants depends critically on these rules for generating counterexamples.  In particular, values in $S$ are added to the set $V_-$ of negative examples in order to strengthen a candidate invariant, while values in $V$ are added to the set $V_+$ of positive examples in order to weaken a candidate invariant.
Therefore, the returned set $S$ ($V$) must be non-empty whenever strengthening
(weakening) is required, which we prove as part of our completeness theorem.

Given this theory of counterexamples, one can appreciate why
handling higher-order functions is more challenging than first-order
functions.  Extracting counterexamples from a pair or other data structure
requires a walk of the data structure, and such a procedure is trivially
complete.  However, extracting counterexamples
from functional arguments requires execution of those arguments.
That said, it is easy to extract counterexamples from
functional arguments when the types of those functions do not include
the abstract type $\alpha$---in that case, there are no counterexamples and one
could safely return the empty set.  Therefore, our theory and formal guarantees
extend naturally to modules that contain
functions such as maps, folds (other than those that produce values of the abstract type), zips, and iterators, where function
argument types refer to the {\em element} type of a data structure,
not the abstract type of the data structure itself.  The latter
case actually appears surprisingly rare in practice, but it does
exist. For instance, the abstract type appears in a
higher-order position in a monadic interface.  We explain how we
lift the first-order restriction in our implementation in \cref{sec:implementation}.

\subsection{The Inference Algorithm}%
\label{sec:the-repinvgen-system.our-inference-algorithm}






The invariant synthesis algorithm is
parameterized by a verifier \verifier\ and a synthesizer \synthesizer.
A call \verifier\ $P$ returns \Valid\ when $P\ v$ is true on all inputs
of type $\tau_c$.  Otherwise, it returns a counterexample $v$ to the predicate.
A call \synthesizer\ \Vpos\ \Vneg\  returns a predicate $P$ that
returns true on the positive examples (\Vpos) and false on the
negative ones (\Vneg).
\Vpos\ and \Vneg\ should not overlap; if they do then \synthesizer\ will fail.

\begin{figure}[!t]\small%
\begin{boxed-listing}[left=-28pt]{}
#\textsf{\emph{Dependencies:} A synthesizer \synthesizer\ and a verifier \verifier}#

#\textsf{\emph{Globals:} An interface $\interface = \istd$, a module $\module = \mstd$ s.t. $\HasType {\mval} {\Subst{\mtyp}{\alpha}{\ctyp}}$, and a spec $\HasType {\spec} {\Forall {\alpha} {\mtyp \to \alpha \to \TBool}}$}\label{line:globals}#

(* The specification $\color{listing-comment}\color{listing-comment}\spec$ interpreted over the concrete type $\color{listing-comment}\ctyp$
 * and module implementation $\color{listing-comment}\mval$ *)
let $\spec_m$ = $\Apply{\spec[\ctyp]}{\mval}$

(* Checks whether a candidate invariant $\color{listing-comment}Q$ is
 * conditionally inductive with respect to $\color{listing-comment}P$ *)
let #\IsCondInductive# $P$ $Q$ = $R$ where $\LeadsTo{P}{Q}{\mval}{\mtyp}{R}$

(* Checks if the candidate invariant is missing any value
 * that is constructible from $\color{listing-comment}\Vpos$ in a single step *)
let #\ClosedPositives# #\Vpos# #\inv# =
  match #\IsCondInductive# #\Vpos# #\inv# with#\label{line:cond-inductiveness-check}#
  | #\Valid# $\mapsto$ #\Valid#
  | #\CounterExample{\Tuple{\_}{V}}# $\mapsto$ $\CounterExample{V}$

(* Checks if the candidate invariant is not inductive,
 * or includes values that are not constructible *)
let #\NoNegatives# #\inv# =
  match #\verifier# #\Sufficient[\spec][\module]{\inv}# with#\label{line:sufficiency-check}#
  | #\Valid# $\mapsto$ begin
      match #\IsCondInductive# #\inv# #\inv# with#\label{line:inductiveness-check}#
      | #\Valid# $\mapsto$ #\Valid#
      | #\CounterExample{\Tuple{S}{\_}}# $\mapsto$ $\CounterExample{S}$
    end
  | $\CounterExample{v}$ $\mapsto$ $\CounterExample{\{v\}}$

(* Returns a sufficient representation invariant *)
let rec #\ALGO{Hanoi}# #\Vpos# #\Vneg# =#\label{line:repinvgen-entry}#
  match #\synthesizer# #\Vpos# #\Vneg# with
  | #\Failure# $\mapsto$ failwith #"No predicate found"#
  | #\Success{\inv}# $\mapsto$ begin
      match #\ClosedPositives# #\Vpos# #\inv# with#\label{line:positive-check}#
      | $\CounterExample{P}$ $\mapsto$ #\textnormal{\ALGO{Hanoi}}# $\left( \Vpos \cup P \right)$ $\emptyset$#\label{line:positive-recurse}#
      | #\Valid# $\mapsto$ begin
          match #\NoNegatives# #\inv# with#\label{line:negative-check}#
          | $\CounterExample{N}$ $\mapsto$
            if $N \setminus \Vpos = \emptyset$ then
              failwith #"Counterexample $N$";#
            else
              #\textnormal{\ALGO{Hanoi}}# #\Vpos# $\left( \Vneg \cup \left(N \setminus \Vpos\right) \right)$#\label{line:negative-recurse}#
          | #\Valid# $\mapsto$ #\inv##\label{line:repinvgen-result}#
        end
    end
\end{boxed-listing}
\caption{The \RepInvGen\ framework.}%
\label{fig:repinvgen-algos}
\end{figure}

\Cref{fig:repinvgen-algos} presents our invariant inference algorithm.  To
execute the algorithm, a user invokes \RepInvGen\ (\cref{line:repinvgen-entry})
with empty sets for \Vpos\ and \Vneg\ respectively.
\RepInvGen\ first generates a candidate invariant \inv\
using \synthesizer, given the current \Vpos\ and \Vneg\ sets.  It then attempts to produce a candidate invariant that is visibly inductive relative to \Vpos.  That is the role of the call to \ClosedPositives\
(\cref{line:positive-check}).  That function simply calls \IsCondInductive, which uses the inference rules in \cref{fig:rules}.  In the implementation these rules are executed through interaction with the verifier \verifier. 

 Since everything in \Vpos is known to be constructible, the set $V$ of values that violate the candidate invariant must also be constructible.  Therefore, those values are returned from \ClosedPositives, and they are added to \Vpos via a recursive call to \RepInvGen\ (\cref{line:positive-recurse}).
This forces future candidate invariants produced by \synthesizer\ to
return \CF{true} on elements in $V$.  Note that each time \Vpos\ is augmented, \Vneg\ is reset
to the empty set, so the next synthesized invariant will be the constant
true function, which is trivially visibly inductive.
While we maintain the invariant that the positive examples are constructible and so must be included in the final invariant, negative examples are simply values that violate the current candidate invariant (but may in fact be constructible).


Once the candidate invariant \inv\ is visibly inductive with respect to \Vpos,
\RepInvGen\ checks for sufficiency and full inductiveness by calling
\NoNegatives\ at \cref{line:negative-check}. The \NoNegatives\ procedure
interacts with \verifier\ to check sufficiency and calls \IsCondInductive\ to
check full inductiveness. If either of these checks fail, \NoNegatives\ will
return counterexample values that satisfy the current invariant --- either a
sufficiency violation or the set $S$ from an inductiveness counterexample.
Because \inv\ is visibly inductive, $N \setminus \Vpos$ can only be empty if
there is a sufficiency violation. In that case, we have found a constructible
violation of the specification \spec, so \RepInvGen\ terminates and provides
this counterexample. If $N \setminus \Vpos$ is non-empty, \RepInvGen\ adds all
of these values to \Vneg, and \synthesizer\ will generate
a stronger candidate invariant in the next iteration. If a constructible
counterexample is added to \Vneg, it will eventually be generated by
\ClosedPositives\ and moved to \Vpos. If both checks in \NoNegatives\ succeed,
then we have found a sufficient representation invariant and it is returned.


\subsection{Soundness and Completeness}%
\label{sec:the-repinvgen-system.properties}


We say that \verifier\ is \emph{sound}
if $\Apply{\verifier}{P} = \Valid$ implies
$\Forall{\HasType{v}{\ctyp}}{\EvaluatesTo{\Apply{p}{v}}{\CF{true}}}$.
Further, \verifier\ is said to be complete if
$\verifier(p) = v$ implies $\EvaluatesTo{\Apply{p}{v}}{\CF{false}}$.
Likewise, we say that \synthesizer\ is \emph{sound} if
for all sets \Vpos\ and \Vneg\ of \ctyp\ values,
$\Apply{\synthesizer}{\Vpos}{\Vneg} = P$ implies
$\Forall{v \in \Vpos}{\EvaluatesTo {\Apply{P}{v}} {\CF{true}}}$ and
$\Forall{v \in \Vneg}{\EvaluatesTo {\Apply{P}{v}} {\CF{false}}}$.
Further, \synthesizer\ is said to be \emph{complete} if
for all sets \Vpos\ and \Vneg\ of \ctyp\ values,
whenever there exists a predicate $\HasType{P}{\ctyp \to \TBool}$ such that
$\Forall{v^{+} \in \Vpos}{\EvaluatesTo {\Apply{P}{v^{+}}} {\CF{true}}}$ and
$\Forall{v^{-} \in \Vneg}{\EvaluatesTo {\Apply{P}{v^{-}}} {\CF{false}}}$,
\synthesizer\ always returns \emph{some} predicate $P'$.



\begin{definition}[Soundness]\label{def:soundness}
  An inference system for representation invariants is said to be sound iff 
  whenever the system generates a predicate \inv, it is indeed a sufficient representation invariant,
  i.e. $\Satisfies{\module}{\interface}{\inv}{\spec}$.
\end{definition}

\begin{definition}[Completeness]\label{def:completeness}
  An inference system for representation invariants is said to be complete iff 
  whenever there exists a sufficient representation invariant \inv{}
  such that $\Satisfies{\module}{\interface}{\inv_\star}{\spec}$,
  the system always generates (terminates with) some predicate $\HasType{\inv}{\TFun{\ctyp}{\TBool}}$.
\end{definition}

\begin{theorem}\label{thm:RepInvGen-sound}
  If \:\verifier\ is sound, then \RepInvGen\ $\emptyset$ $\emptyset$ is sound.
 \end{theorem}


\begin{theorem}\label{thm:RepInvGen-complete-finite}
  If \,\verifier\ and \,\synthesizer\ are both sound and complete, and \ctyp\ is
  a finite domain,
  then \RepInvGen\ $\emptyset$ $\emptyset$ is complete.
\end{theorem}

\ifappendix Please see the appendix for proofs. \else The proofs can be found in
the appendix of the full version of this paper~\cite{todo-arxiv}. \fi  The soundness of \RepInvGen\ is straightforward and follows from the fact that an invariant is only returned if it is both sufficient and inductive.  The completeness argument for finite domains is much more involved.  As mentioned earlier, it depends on several properties of the rules for generating counterexamples in \cref{fig:rules}.  Further, we must prove that the \RepInvGen\ algorithm always terminates.  Notice that the size of the set \Vpos\ monotonically increases during the algorithm.  While \Vneg\ is reset to empty on some recursive calls, this is only done  when \Vpos\ is augmented.   Hence the following is a rank function that is bounded from below and decreases lexicographically with each recursive \RepInvGen\ call, where \card{\ctyp} denotes the number of values of type \ctyp:  $$ \rank(\Vpos, \Vneg) \definedas \Tuple{\card{\ctyp} - \card{\Vpos}}{\card{\ctyp} - \card{\Vneg}} $$



\section{Implementation}%
\label{sec:implementation}
This section describes a variety of additional aspects of our \textasciitilde{}5 KLOC OCaml implementation of
\Hanoi{}.

\subsection{The Programming Language}
We have implemented a pure, simply-typed, call-by-value functional language with
recursive data types. Numbers are implemented as a recursive data type, where a
number is either $0$ or the successor of a number. Each program includes a
prelude that may contain data type declarations and functions over those data
types. A program also contains a single module declaring an abstract type
together with operations over that abstract type. Finally, a program includes a
universally quantified specification that defines the intended behavior of the
module in terms of its operations.

\subsection{Tackling Higher-Order Functions}
\label{subsec:hofs}


While the theory presented in the previous section only supports first-order terms, our implementation allows modules to include arbitrary higher-order functions.  As mentioned earlier, the key extension required is the ability to extract counterexample values of the abstract type from functions.  Here we discuss how our implementation does that.

First, consider a
natural extension to the SET interface from
\cref{sec:a-motivating-example} to include a map function.

\begin{lstlistingsmall}
module type HOSET = sig
   include SET
   val map : (int -> int) -> t -> t
end
\end{lstlistingsmall}

\noindent
Notice that while map is a higher-order function, the type of the
higher-order argument does not involve any occurrences of the abstract
type \texttt{t}.  The same is true of iter, zip, and many other variants.
Consequently, if, during invariant inference,
(\texttt{map}\, $f$\, $v$) fails an inductiveness check on some candidate representation invariant, the counterexample values that represent the ``reason'' for this failure will never come from $f$. More generally, when the abstract type does not appear in a higher type
$\tau$, the value with type $\tau$ cannot contain counterexamples.  Our implementation therefore simply ignores such higher-order values when extracting counterexamples,
just as it ignores ordinary base types such \texttt{int}.







Now consider a further extension that includes a fold. 

\begin{lstlistingsmall}
module type FSET = sig
   include HOSET
   val fold : (int -> t -> t) -> t -> t -> t
end
\end{lstlistingsmall}
\noindent
Here, \texttt{fold} contains a function argument with a type including
\texttt{t}. The fold might be implemented as follows.

\begin{lstlistingsmall}
let rec fold f a s =
   match s with
   | [] -> a
   | hd :: tl -> f hd (fold f a tl)
\end{lstlistingsmall}

Given a call \texttt{fold f s1 s2} and a result \texttt{s$'$} that
does not satisfy the current candidate invariant, how do we extract
the counterexamples from the functional argument?  The solution
arises from reflecting back on the intuitive definition of conditional inductiveness:  ``if
clients supply values in $P$ then the module implementation should
supply values in $Q$.''  In the higher-order case, there are simply
more ways for client and implementation to interact across the module
boundary.  Specifically,the implementation supplies a value to the client when it
calls a function argument, and the client supplies a value to the module when it returns from such a function.
Fortunately, a mechanism already exists for tracking such boundary
crossings in the general case:  the higher-order contracts of Findler
and Felleisen~\cite{Findler02}.

Therefore, our implementation extracts counterexamples through higher-order contract checking.
The first-order case is straightforward.  For example, when the type is \texttt{t -> t}, we generate a contract
\texttt{$P$ -> $Q$} to check that arguments satisfy $P$ and results satisfy $Q$, and we log situations where $P$ is satisfied by an argument but $Q$ is violated by the result.
This is a direct implementation of the rule \ref{rule:i-fun-cex} in \cref{fig:rules}.

For a type such as \texttt{(int -> t -> t) -> t -> t -> t}, we simply
extend the idea, giving rise to the following contract.
\[
  \texttt{(any\_int -> $Q$ -> $P$)  -> $P$ -> $P$ -> $Q$}
\]
\noindent
As per usual, all negative
positions must satisfy $P$ and the positive ones $Q$.  Then contract checking is used to identify runs that satisfy all of the $P$ checks but fail a $Q$ check.  In that case, if $S$ is the set of values that satisfy $P$ and $v$ is the value that violates $Q$, then the extracted inductiveness counterexample is
$\Tuple{S}{\{v\}}$.

With this extension, \RepInvGen\ is trivially sound, for the same reason that the
first-order algorithm is sound (the algorithm checks for soundness
just before termination).  We conjecture that \RepInvGen\ with this extension is also complete for finite domains but have not proven it.  However, in the next section, we demonstrate empirically that our implementation can infer representation invariants in the presence of higher-order functions.


%
%
%





\subsection{Verifier and Synthesizer}
\label{sec:impl-synth}

To implement $\verifier$, we use a size-bounded
enumerative tester, which is unsound but effective in practice.
To validate a predicate with a single quantifier, we test the predicate on data
structures, from smallest to largest, until either 3000 data structures have
been processed, or the data structure has over 30 AST nodes, whichever comes
first. To validate predicates with two or more quantifiers, we instantiate each
quantifier with the smallest 3000 data structures with under 15 AST nodes. We
further limit the total number of data structures processed to 30000. These
restrictions limit the total amount of time spent in the verifier at the cost of
soundness guarantees.

To implement $\synthesizer$, we use \Myth~\cite{pldi15/osera/myth},
adapting it slightly in two ways. First, we modified it to return a set of
candidate invariants, instead of just one. Doing so permits the caching of
synthesis results described earlier. Second, we had to manage \Myth's
requirement for \emph{trace completeness}. Trace completeness requires that
whenever we provide an input-output example \Tuple{x}{y}\ for a recursive data
type, we also provide input-output examples for each subvalue of $x$. We
generate input-output pairs for \Myth\ by pairing each element of \Vpos\ with
true and each element of \Vneg\ with false. To handle trace completeness, we
first identify all subvalues of the values in \Vpos\ and \Vneg. For each such
subvalue that does not already appear in \Vpos\ or \Vneg\, we simply add it to
\Vneg, which has the effect of mapping it to false. However, it could be that
these values are actually constructible; if they are, future visible
inductiveness checks will find this inconsistency, and move the value to \Vpos.
However, this solution sometimes does make our synthesis task more difficult, as
these additional values in \Vneg\ can force candidate invariants to be stronger than necessary.  In such cases the synthesizer can spend more time searching for a complex invariant, when a simpler (though weaker) one would suffice.

\subsection{Optimizations}
\label{subsec:alg}
To accelerate invariant inference, we have implemented two key
optimizations:
\textit{synthesis result caching} and \textit{counterexample list
  caching}.
Synthesis result caching
reduces the number of synthesis calls, and counterexample list caching
reduces the number of verification and synthesis calls.
Since the bulk of the system run time is spent in one or both kinds of calls,
reducing them can have a substantial impact on performance.

\paragraph*{Synthesis Result Caching.} When synthesizing, \Myth\ often finds
multiple possible solutions for a given set of input/output examples. Instead of
throwing the unchosen solutions away, we store them for future synthesis calls.
When given a set of input/output examples, before making a call to \Myth{}, we
check if any of the previously synthesized invariants satisfy the input/output example
set. If one does, that invariant is used instead of a freshly synthesized one.

\paragraph*{Counterexample List Caching}
\begin{figure}
  \centering
  \subfigure[]
    {\label{subfig:initial-result-list}
    \includegraphics[scale=.3]{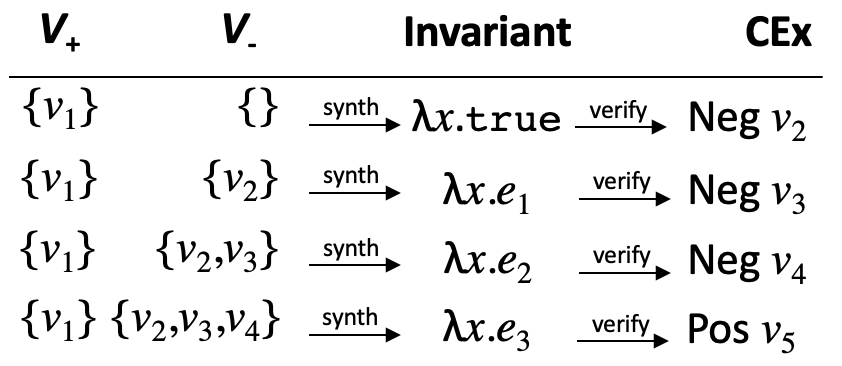}}
  \subfigure[]
    {\label{subfig:result-filter}
    \includegraphics[scale=.3]{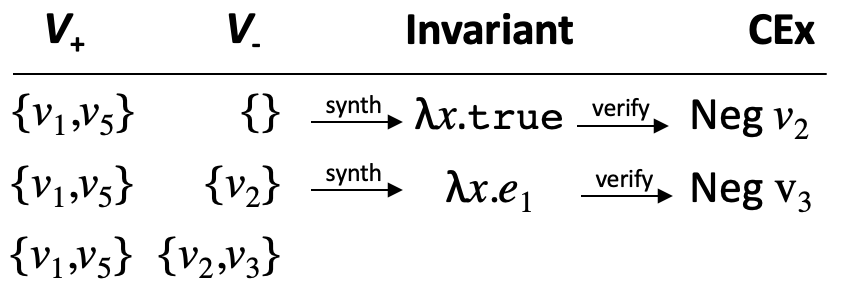}}
  \caption{Partial traces of synthesis and verification results from a run
    of \RepInvGen\ without counterexample list caching enabled.}
\end{figure}
Consider the trace of \Hanoi\ shown in \cref{subfig:initial-result-list}.
In this example, \RepInvGen\ was just called with $v_1$ as the only positive
example, and with no negative examples. With no negative examples,
the synthesizer proposes $\lambda x.\CF{true}$ as a candidate
invariant and then verification subsequently provides the negative
counterexample, $v_2$, which
then becomes the only negative example in the next attempt at
synthesis. This loop of proposing new
invariants, and adding their negative counterexamples to the negative example
set continues until $\lambda x.e_3$ is proposed, which provides the positive
counterexample, $v_5$. 

\begin{figure}
    \includegraphics[scale=.3]{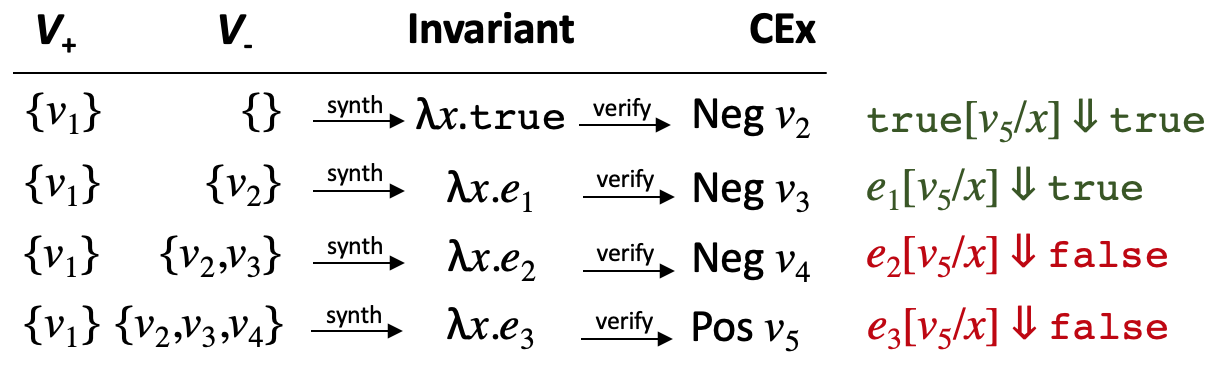}
    \caption{Results of running a positive counterexample on the trace shown in
      \cref{subfig:initial-result-list}.}
    \label{fig:filter}
\end{figure}

Next, according to the unoptimized algorithm, one should begin a run
with $\{v_1,v_5\}$ as positive examples and no negative examples---see
\cref{subfig:result-filter} for a partial trace of this
subsequent execution. Suppose that $v_5$ satisfies the first two synthesized invariants from the original run.  In that case, those invariants will simply be synthesized again as the first two candidates of this new run, as shown in the figure.  To avoid this
recomputation, we cache these traces of synthesis and verification calls. When
we receive a new positive example, we run it on the synthesized invariants in
the trace, as shown in \cref{fig:filter}. Because $\lambda x.\CF{true}$
and $\lambda x.e_1$ both return \CF{true} on $v_5$, we can skip the entirety of the
trace shown in \cref{subfig:result-filter} and begin by synthesizing from
the testbed $\Vpos = \{v_1,v_5\}$ and $\Vneg = \{v_2,v_3\}$.


\section{Experimental Results}%
\label{sec:experiments}

We aim to answer the following research questions:

\begin{enumerate}
\item Can we infer sufficient representation invariants in practice?
\item What are the primary performance factors?
\item What effect on performance do our optimizations have?
\item How does our algorithm compare with prior work?
\end{enumerate}

\subsection{Benchmark Suite}


We evaluate \RepInvGen\ on a total of \NumBenchmarks\ verification problems,
most of which require reasoning over list or tree structures.
We categorize them into the following four groups.
\begin{itemize}[topsep=0.5em,itemsep=0.25em]
\item \fancyname{VFA}~($5$):
  Four modules from Verified Functional
  Algorithms (VFA)~\citep{web18/appel/vfa} that have interfaces and
  specifications over those interfaces, including
  tree- and list-based implementations of lookup
  tables and priority queues.  We
  also experimented with a
  second version of priority queues that excludes the \CF{merge} function.
\item \fancyname{VFAExt}~($3$):
  Three \fancyname{VFA} modules with additional function(s)
  and corresponding specifications from the Coq~\citep{coq} standard library.
\item \fancyname{Coq}~($14$):
  Five tree- and list-based implementations of
  data structures from the Coq~\citep{coq} standard library.
  One additional problem for each of the five
  by introducing additional binary functions.  Four more problems by
  extending interfaces with higher-order functions.
\item \fancyname{Other}~($6$): Six additional benchmarks of our own creation requiring
  reasoning over lists, natural numbers, monads and other basic data structures. 
\end{itemize}


\subsection{Experimental Setup}

All experiments were performed on a 2.5 GHz Intel Core i7 processor with 16 GB
of 1600 MHz DDR3 RAM running macOS Mojave. We ran each benchmark 10 times with a
timeout of 30 minutes and report the average time. If any of the 10 runs time
out then we consider the benchmark as a whole to have timed out.

\subsection{Inferred Invariants}

\cref{fig:running-times} presents our results.  Overall,
\RepInvGen\ terminated with an invariant on \NumSuccess\ out of
\NumBenchmarks\ benchmarks within the timeout bound. The second
column shows the sizes of the inferred
invariants, in terms of their abstract syntax trees.

Though our verifier is unsound, there was no effect on the reliability
of the system on
our benchmark suite:  \NumSuccess\ of the \NumSuccess\ inferred invariants are
correct.  Further, some of them are quite sophisticated.  For example, we
synthesize a heap invariant, which requires that the the elements of each node's
subtrees is smaller than that node's label.  We synthesize invariants over lists
including ``max element first,'' ``no duplicates,'' and
``ordered.''  If we allow the system to exceed the 30 min threshold, the
system will infer a binary search tree invariant as well.

In seven of the cases above, we run into a limitation of the \Myth\
synthesizer rather than our algorithm:  \Myth\ cannot synthesize
functions that require recursive ``helper'' functions.
To bypass this
restriction, we added a \CF{true_maximum} function (that finds the maximum
element of a tree) to our tree-heap benchmark and a \CF{min_max_tree} function
(that finds the minimum and maximum elements of trees) to our bst and
red-black-tree benchmarks.  We
added a * next to the names of benchmarks that we altered by providing
a helper function in this way (see \cref{fig:running-times}).

\subsection{Primary Performance Factors}

\begin{figure*}%
  \small
  \centering%
  \begin{tabulary}{\textwidth}{R|C|C|C|C|C|C|C|C}
    \Name & \Size & \Time\ (s) & \TVT\ (s) & \TVC & \MVT\ (s) & \TST\ (s) & \TSC & \MST\ (s)\\
    \hline \hline \begin{filecontents*}{generated-data-annotated.csv}
MeanVerifTime,InvariantSize,MythVerifCalls,FoldMaxVerifTime,MythMaxVerifTime,ThirtyTotalVerifTime,FoldTotalSynthTime,ConjStrVerifCalls,Test,ThirtyTime,SRPVerifCalls,LAStrTotalSynthTime,MythTime,SRPMaxSynthTime,FoldSynthCalls,LAStrMaxSynthTime,ThirtySynthCalls,SRPTotalSynthTime,CLPTotalVerifTime,MythTotalSynthTime,ThirtyMaxSynthTime,FoldTime,MythSynthCalls,SRPTime,SRPMaxVerifTime,ConjStrTotalVerifTime,LAStrSynthCalls,ConjStrSynthCalls,CLPMaxVerifTime,LAStrTime,ThirtyMaxVerifTime,MythTotalVerifTime,CLPTotalSynthTime,CLPTime,ConjStrTotalSynthTime,CLPSynthCalls,CLPVerifCalls,LAStrTotalVerifTime,ConjStrTime,FoldTotalVerifTime,FoldVerifCalls,ThirtyTotalSynthTime,ConjStrMaxSynthTime,ThirtyVerifCalls,ConjStrMaxVerifTime,MythMaxSynthTime,FoldMaxSynthTime,SRPSynthCalls,MeanSynthTime,LAStrVerifCalls,CLPMaxSynthTime,LAStrMaxVerifTime,SRPTotalVerifTime
97,t/o,t/o,t/o,t/o,incorrect,t/o,t/o,/coq/bst-::-set*,incorrect,t/o,t/o,t/o,t/o,t/o,t/o,incorrect,t/o,t/o,t/o,incorrect,t/o,t/o,t/o,t/o,t/o,t/o,t/o,t/o,t/o,incorrect,t/o,t/o,t/o,t/o,t/o,t/o,t/o,t/o,t/o,t/o,incorrect,t/o,incorrect,t/o,t/o,t/o,t/o,t/o,t/o,t/o,t/o,t/o
3.80,15,11,21.3,17.8,err,2.0,66,/coq/bst-::-set+binfuncs,err,11,0.3,42.0,0.1,7,0.2,err,0.2,41.3,0.2,err,55.7,3,40.7,16.3,76.3,4,7,16.7,114.3,err,41.8,0.2,41.5,0.2,3,12,114.0,76.6,53.5,20,err,0.1,err,16.5,0.1,0.4,3,0.07,27,0.1,23.9,40.5
97,t/o,t/o,t/o,t/o,incorrect,t/o,t/o,/coq/bst-::-set+hofs*,incorrect,t/o,t/o,t/o,t/o,t/o,t/o,incorrect,t/o,t/o,t/o,incorrect,t/o,t/o,t/o,t/o,t/o,t/o,t/o,t/o,t/o,incorrect,t/o,t/o,t/o,t/o,t/o,t/o,t/o,t/o,t/o,t/o,incorrect,t/o,incorrect,t/o,t/o,t/o,t/o,t/o,t/o,t/o,t/o,t/o
103,t/o,t/o,t/o,t/o,incorrect,t/o,t/o,/coq/rbtree-::-set*,incorrect,t/o,t/o,t/o,t/o,t/o,t/o,incorrect,t/o,t/o,t/o,incorrect,t/o,t/o,t/o,t/o,t/o,t/o,t/o,t/o,t/o,incorrect,t/o,t/o,t/o,t/o,t/o,t/o,t/o,t/o,t/o,t/o,incorrect,t/o,incorrect,t/o,t/o,t/o,t/o,t/o,t/o,t/o,t/o,t/o
103,t/o,t/o,t/o,t/o,err,t/o,t/o,/coq/rbtree-::-set+binfuncs,err,t/o,t/o,t/o,t/o,t/o,t/o,err,t/o,t/o,t/o,err,t/o,t/o,t/o,t/o,t/o,t/o,t/o,t/o,t/o,err,t/o,t/o,t/o,t/o,t/o,t/o,t/o,t/o,t/o,t/o,err,t/o,err,t/o,t/o,t/o,t/o,t/o,t/o,t/o,t/o,t/o
103,t/o,t/o,t/o,t/o,incorrect,t/o,t/o,/coq/rbtree-::-set+hofs*,incorrect,t/o,t/o,t/o,t/o,t/o,t/o,incorrect,t/o,t/o,t/o,incorrect,t/o,t/o,t/o,t/o,t/o,t/o,t/o,t/o,t/o,incorrect,t/o,t/o,t/o,t/o,t/o,t/o,t/o,t/o,t/o,t/o,incorrect,t/o,incorrect,t/o,t/o,t/o,t/o,t/o,t/o,t/o,t/o,t/o
0.36,35,16,6.2,4.2,incorrect,2.7,56,/coq/maxfirst-list-::-heap,incorrect,18,0.6,6.2,0.1,15,0.2,incorrect,0.3,5.9,0.3,incorrect,13.5,5,6.6,4.4,8.5,15,15,4.2,18.1,incorrect,5.8,0.3,6.3,0.4,5,21,17.5,9.1,10.0,45,incorrect,0.2,incorrect,3.8,0.1,0.6,7,0.06,67,0.1,4.1,6.2
0.44,35,16,6.3,4.3,err,1.3,58,/coq/maxfirst-list-::-heap+binfuncs,err,18,0.7,7.4,0.2,12,0.3,err,0.3,7.2,0.3,err,17.6,5,8.6,4.5,16.4,15,15,4.3,30.2,err,7.0,0.3,7.6,0.4,5,21,29.5,17.0,15.7,40,err,0.2,err,3.8,0.2,0.2,7,0.06,69,0.2,4.3,8.2
1.05,49,21,2.7,13.1,incorrect,0.3,127,/coq/sorted-list-::-set,incorrect,21,1.7,22.9,0.5,4,0.5,incorrect,0.8,24.7,0.7,incorrect,6.7,7,22.7,12.9,114.3,19,20,14.1,51.6,incorrect,22.0,0.7,25.5,1.7,7,30,49.9,116.2,6.1,16,incorrect,0.5,incorrect,15.2,0.5,0.1,8,0.10,59,0.6,15.6,21.8
0.76,49,21,5.0,5.5,err,0.4,50,/coq/sorted-list-::-set+binfuncs,err,21,2.1,17.3,0.7,4,0.7,err,1.2,42.1,1.2,err,10.8,8,16.6,5.3,19.8,17,6,5.5,48.5,err,16.0,1.1,43.3,0.1,8,42,46.4,20.0,10.2,16,err,0.1,err,5.4,0.7,0.1,8,0.15,58,0.7,6.0,15.3
4.74,49,21,41.9,62.3,incorrect,0.5,t/o,/coq/sorted-list-::-set+hofs,incorrect,21,4.2,101.3,1.1,4,1.5,incorrect,1.7,t/o,1.6,incorrect,47.5,8,116.8,73.5,t/o,15,t/o,t/o,344.7,incorrect,99.6,t/o,t/o,t/o,t/o,t/o,340.4,t/o,46.7,16,incorrect,t/o,incorrect,t/o,1.0,0.1,8,0.20,62,t/o,86.7,115.0
0.65,35,20,12.8,10.0,0.0,16.5,37,/coq/unique-list-::-set,1.2,20,0.4,13.2,0.1,9,0.1,0,0.4,12.9,0.3,0.0,32.8,6,13.3,10.0,15.2,8,11,10.0,27.5,0.0,12.9,0.3,13.3,0.3,6,25,27.0,15.7,15.8,26,0.0,0.2,0,6.1,0.1,15.1,8,0.05,37,0.1,6.0,12.8
1.95,15,8,13.6,11.8,err,17.5,25,/coq/unique-list-::-set+binfuncs,err,8,0.1,15.7,0.1,9,0.1,err,0.1,15.7,0.1,err,37.1,2,15.7,11.8,16.3,3,5,11.9,32.0,err,15.6,0.1,15.8,0.1,2,9,31.8,16.5,19.1,30,err,0.1,err,7.2,0.1,15.4,2,0.05,26,0.1,7.2,15.5
10.18,17,8,81.6,78.2,incorrect,2.8,25,/coq/unique-list-::-set+hofs,incorrect,8,0.2,81.7,0.2,9,0.2,incorrect,0.2,70.4,0.2,incorrect,92.1,2,67.4,64.2,69.0,3,5,67.5,138.8,incorrect,81.4,0.2,70.7,0.2,2,9,138.6,69.2,88.8,28,incorrect,0.2,incorrect,50.4,0.2,1.0,2,0.10,26,0.2,52.2,67.1
0.06,29,16,0.6,0.5,incorrect,0.4,50,/other/cache,incorrect,16,0.2,1.3,0.1,3,0.1,incorrect,0.2,0.9,0.2,incorrect,1.5,5,1.1,0.5,1.0,8,8,0.4,2.1,incorrect,1.0,0.1,1.2,0.2,5,23,1.8,1.2,1.0,14,incorrect,0.1,incorrect,0.3,0.1,0.2,5,0.04,37,0.1,0.3,0.9
0.63,53,14,6.3,4.7,incorrect,0.9,33,/other/listlike-tree,incorrect,14,0.1,9.0,0.1,6,0.1,incorrect,0.1,8.8,0.1,incorrect,19.0,4,8.5,4.5,25.6,9,10,4.4,33.9,incorrect,8.8,0.1,9.0,0.1,4,19,33.7,25.8,18.0,16,incorrect,0.1,incorrect,4.5,0.1,0.2,5,0.03,35,0.1,4.4,8.3
0.12,23,12,0.9,1.0,t/o,0.2,31,/other/nat-nat-option-::-range,t/o,12,0.1,1.6,0.1,3,0.1,t/o,0.1,1.5,0.1,t/o,1.4,4,1.7,1.0,2.5,5,9,1.0,4.2,t/o,1.5,0.1,1.6,0.1,4,15,4.0,2.7,1.2,10,t/o,0.1,t/o,1.0,0.1,0.1,4,0.03,32,0.1,1.0,1.5
0.97,28,8,6.7,6.8,incorrect,0.6,20,/other/rational,incorrect,8,0.7,8.6,0.7,5,0.6,incorrect,0.7,7.8,0.7,incorrect,8.8,2,8.6,6.7,7.7,3,5,6.7,15.8,incorrect,7.8,0.7,8.6,0.7,2,9,15.1,8.4,8.1,14,incorrect,0.6,incorrect,6.6,0.7,0.2,2,0.35,21,0.7,6.9,7.8
0.75,45,20,6.4,11.8,incorrect,0.4,42,/other/sized-list,incorrect,18,0.3,15.4,0.2,4,0.2,incorrect,0.3,16.2,0.3,incorrect,11.4,6,15.8,12.1,24.5,14,17,12.1,39.4,incorrect,15.0,0.3,16.7,0.4,6,33,39.0,25.0,10.8,14,incorrect,0.2,incorrect,6.3,0.2,0.2,7,0.05,43,0.2,6.5,15.3
0.28,49,20,2.1,2.1,err,0.2,35,/other/stutter-list,err,20,1.1,6.9,0.5,2,0.4,err,1.1,9.6,1.1,err,4.2,7,6.9,2.1,9.4,9,10,2.1,13.9,err,5.6,1.2,11.0,1.1,7,34,12.8,10.6,3.9,10,err,0.4,err,2.1,0.5,0.1,8,0.16,37,0.5,2.2,5.6
0.83,4,3,2.4,2.5,incorrect,0.0,11,/vfa-extended/assoc-list-::-table,incorrect,3,0.0,2.6,0.0,0,0.0,incorrect,0.0,2.5,0.0,incorrect,2.5,0,2.5,2.4,57.2,0,0,2.4,54.9,incorrect,2.5,0.0,2.5,0.0,0,3,54.9,57.3,2.4,3,incorrect,0.0,incorrect,51.8,0.0,0.0,0,undef,6,0.0,49.6,2.4
103,t/o,t/o,t/o,t/o,incorrect,t/o,t/o,/vfa-extended/bst-::-table,incorrect,t/o,t/o,t/o,t/o,t/o,t/o,incorrect,t/o,t/o,t/o,incorrect,t/o,t/o,t/o,t/o,t/o,t/o,t/o,t/o,t/o,incorrect,t/o,t/o,t/o,t/o,t/o,t/o,t/o,t/o,t/o,t/o,incorrect,t/o,incorrect,t/o,t/o,t/o,t/o,t/o,t/o,t/o,t/o,t/o
5.13,4,3,7.6,8.0,incorrect,0.0,t/o,/vfa-extended/trie-::-table,incorrect,3,t/o,15.5,0.0,0,t/o,incorrect,0.0,13.4,0.0,incorrect,14.6,0,14.0,7.2,t/o,t/o,t/o,7.0,t/o,incorrect,15.4,0.0,13.5,t/o,0,3,t/o,t/o,14.6,3,incorrect,t/o,incorrect,t/o,0.0,0.0,0,undef,t/o,0.0,t/o,13.9
0.63,4,3,1.9,1.8,incorrect,0.0,9,/vfa/assoc-list-::-table,incorrect,3,0.0,1.9,0.0,0,0.0,incorrect,0.0,1.8,0.0,incorrect,1.9,0,1.9,1.8,51.5,0,0,1.8,51.9,incorrect,1.9,0.0,1.9,0.0,0,3,51.9,51.6,1.9,3,incorrect,0.0,incorrect,49.6,0.0,0.0,0,undef,5,0.0,50.0,1.8
4.27,4,3,12.5,12.8,incorrect,0.0,9,/vfa/bst-::-table,incorrect,3,0.0,12.9,0.0,0,0.0,incorrect,0.0,12.2,0.0,incorrect,12.6,0,12.4,12.3,153.7,0,0,12.1,154.4,incorrect,12.8,0.0,12.2,0.0,0,3,154.3,153.8,12.5,3,incorrect,0.0,incorrect,141.7,0.0,0.0,0,undef,5,0.0,142.4,12.3
0.99,47,54,11.8,13.8,incorrect,34.4,585,/vfa/tree-::-priqueue*,incorrect,52,t/o,65.7,5.8,27,t/o,incorrect,15.8,275.0,9.6,incorrect,81.2,15,76.1,15.1,463.6,t/o,50,13.8,t/o,incorrect,53.6,9.6,287.1,27.3,15,145,t/o,492.9,43.9,68,incorrect,9.0,incorrect,14.4,5.0,11.5,23,0.64,t/o,5.0,t/o,55.6
1.19,47,54,10.7,14.8,err,14.5,t/o,/vfa/tree-::-priqueue+binfuncs*,err,52,t/o,79.4,7.2,21,t/o,err,20.4,335.5,12.4,err,56.2,15,87.4,14.5,t/o,t/o,t/o,14.6,t/o,err,64.5,12.3,350.4,t/o,15,145,t/o,t/o,40.2,53,err,t/o,err,t/o,6.5,3.4,23,0.83,t/o,6.5,t/o,62.4
5.87,4,3,13.2,10.1,incorrect,0.0,t/o,/vfa/trie-::-table,incorrect,3,t/o,17.7,0.0,0,t/o,incorrect,0.0,21.6,0.0,incorrect,22.9,0,26.3,15.1,t/o,t/o,t/o,12.4,t/o,incorrect,17.6,0.0,21.7,t/o,0,3,t/o,t/o,22.9,3,incorrect,t/o,incorrect,t/o,0.0,0.0,0,undef,t/o,0.0,t/o,26.2
\end{filecontents*}
\csvreader[head to column names]{generated-data-annotated.csv}{}
    {\Test & \InvariantSize & \MythTime & \MythTotalVerifTime & \MythVerifCalls & \MeanVerifTime &
      \MythTotalSynthTime & \MythSynthCalls & \MeanSynthTime\\}
  \end{tabulary}
  \caption{Information from running \Hanoi\ on our benchmark suite. \Name\ is
    the name of the benchmark. \Size\ is the size of the inferred invariant.
    \Time\ is the time to run the benchmark from start to end. \TVT\ is the
    total time spent verifying. \TVC\ is the total number of verification calls.
    \MVT\ is the average time for a single verification call. \TST\ is the total
    time spent synthesizing. \TSC\ is the total number of synthesis calls. \MST\
    is the average time for a single synthesis call. Benchmarks marked with a *
    were provided an additional function to enable synthesis by \Myth.}%
  \label{fig:running-times}
\end{figure*}


When benchmarks complete within the 30 minute bound, most of the time is
spent in verification. Indeed, for
all but two of the terminating benchmarks, the total time spent synthesizing is
under two seconds.

Three factors affect verification times significantly: (1) the strength of the
specification, (2) the complexity of the underlying data structure, and (3) the
presence of higher order functions.
First, it takes our verifier longer to validate a true fact than to find a
counterexample to a false one (validation requires enumeration of all
tests; in contrast, the moment a counterexample is found, the
enumeration is short-circuited).
Many candidate invariants imply weak specifications but are
not inductive. Hence weak specifications, ironically, are quite
costly, because many candidate invariants wind up being sufficient (incurring a significant verification expense each
time), only to be thrown away later when it turns out they are not
inductive.
%
%
Second, relatively simple data structures, like natural numbers and lists
with numbers as their elements, take less time to verify than more
complex data structures, like trees, tries, and lists with more complex
elements. 
Third, like other complicated data types, the use of higher-order functions
increases verification time. There are many ways to build a function, so
enumeratively verifying a higher-order function requires searching through many
possible functions.

However, the story is different for the complex benchmarks that do not complete
within 30 minutes. When we ran \Hanoi\ on our bst set benchmark, it completed in
78.4 minutes. Unlike the prior benchmarks, the majority of the time
(65\%) was spent in
synthesis. Moreover, 30\% of the total time was spent on the synthesis call
that generated the final invariant.
This indicates that \Hanoi\ is
currently not gated by the verifier, but by the synthesizer. Indeed, the
implementation of bst set that includes binary functions like union and
intersection is actually much faster than that without union and intersection,
terminating within our 30 minute timeout. Adding these functions makes the
verification harder, but \Myth\ can use them to generate simpler invariants. Adding
helper functions that permit simpler invariants also makes our implementation of
a bst table verifiable in under 30 minutes.


Due to these limitations, we believe that a smarter synthesizer would be able to
find more invariants. To this end, we built a prototype synthesizer that can
generate more complex types of functions. This synthesizer has similarities
to \Myth\ as it is type-and-example directed and enumerative. However, where \Myth\
can only synthesize simple recursive functions, this alternate synthesizer can
synthesize \emph{folds}, letting our synthesizer generate functions that
require accumulators.
Our synthesizer performs
comparably to \Myth, synthesizing invariants for the 20 benchmarks \Myth\ that can solve an
average of 11\% slower. However, our synthesizer is also able to find the invariant
for a binary heap (/vfa/tree-::-priqueue) without requiring helper functions or
functions defined in the module in 185.4 seconds (55.8
seconds for /vfa/tree-::-priqueue+binfuncs), while \Myth\ fails.

\subsection{Comparisons}
\label{sec:comparisons}

\begin{figure}
    \includegraphics[scale=.87]{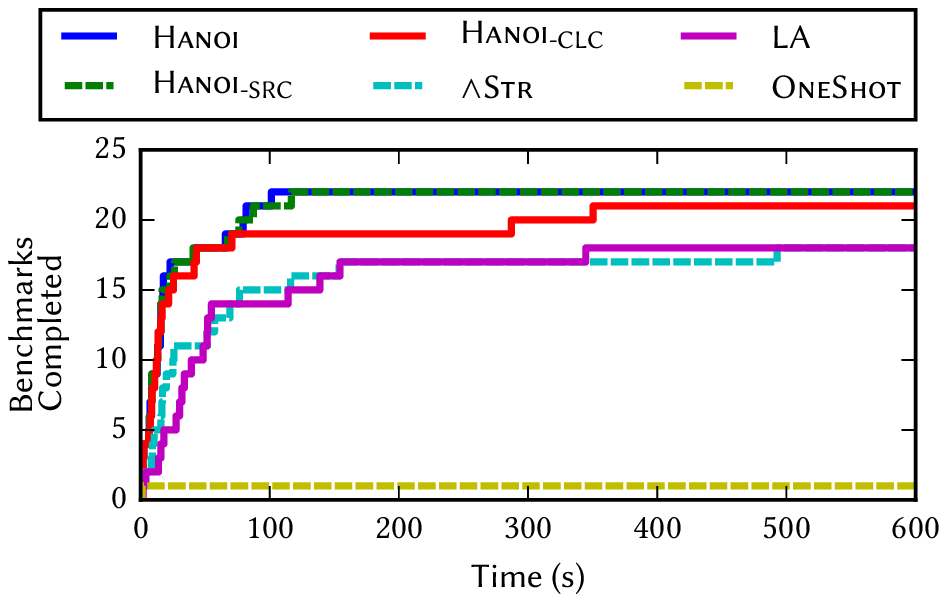}
    \caption{Number of benchmarks that terminate in a given time. \Hanoi\ is the
      full \Hanoi\ tool. \HanoiNoSRC\ is \Hanoi\ with synthesis result caching
      turned off. \HanoiNoCLC\ is \Hanoi\ with counterexample list caching
      turned off. \AndStrengthen\ is \Hanoi\ using a conjunctive strengthening
      algorithm similar to that of \LoopInvGen. \LA\ is \Hanoi\ using a
      counterexample generation strategy similar to that of \LinearArbitrary.
      \OneShot\ synthesizes based on the results of running the 30 smallest
      values on the specification. All the tests that terminated within the 30
      minute timeout terminated within the first 10.}
    \label{fig:times}
\end{figure}
\cref{fig:times} summarizes the results of running of \Hanoi, \Hanoi\ without
optimizations, and our implementations of prior approaches adapted to our setting. 

\paragraph*{Impact of Optimizations.}
The modes \HanoiNoSRC\ and \HanoiNoCLC\ tested the impacts of our optimizations
described in \S\ref{subsec:alg}.
\HanoiNoSRC\ runs the benchmarks with synthesis result caching turned off.
\HanoiNoCLC\ runs the benchmarks with counterexample list caching turned off.

Removing synthesis result caching does not have a large impact on the majority
of benchmarks
as the majority of our benchmarks spend relatively little time in
synthesis. However, more complex benchmarks are able to enjoy the benefits of
this optimization.

Counterexample list caching
has 
significant impact on complex benchmarks
as they have
more synthesis and verification calls. 
The synthesizer requires more input/output
examples to synthesize the correct invariant on complex benchmarks, so saving
time reconstructing the negative examples via counterexample list
caching
has great impact.

\paragraph*{Comparison to \AndStrengthen.} The \AndStrengthen\ mode
simulates the \LoopInvGen{} algorithm~\cite{pldi16/padhi/loopinvgen}, a
related data-driven system for inferring loop invariants.
When running
\AndStrengthen, if a candidate invariant $\inv_1$ is sufficient to prove the
specification, but is not inductive, the algorithm attempts to synthesize a new predicate $\inv_2$ such that
$\LeadsTo{\inv_1 \wedge \inv_2}{\inv_1}{\mval}{\mtyp}{\Valid}$.  In that case, $\inv_1 \wedge \inv_2$ is considered the new candidate invariant.  This process continues until either the conjoined
invariants are inductive, or they are overly strong so a new positive
counterexample is found, at which point the whole process restarts. \Hanoi\ outperforms \AndStrengthen\ on all the
benchmarks and solves 3 more benchmarks within 30 minutes. The main 
downside of \AndStrengthen\ is that it can only add new positive examples in order to weaken the candidate invariant after it has obviously over-strengthened.  \Hanoi, in contrast, uses visible inductiveness checks to eagerly weaken in a directed manner.

\paragraph*{Comparison to \LA.} \LA\ mode simulates the
\LinearArbitrary\ algorithm~\cite{pldi18/zhu/data-driven-chc}, which is
used in a data-driven CHC solver. There are two differences from
\RepInvGen.  First,
\LA\ tries to satisfy individual inductiveness constraints, generated for each function in the module, one at a time rather than all at once.  Second, rather than eagerly searching for visible inductiveness violations, only full inductiveness counterexamples are obtained.  However, if a full inductiveness counterexample happens to also be a visible inductiveness counterexample then it is treated accordingly. \Hanoi\
outperforms \LA\ on all the benchmarks and is able to solve 4
more benchmarks within 30 minutes. While \Hanoi\ checks eagerly for positive
counterexamples, \LA\ finds them nondeterministically. Without performing the
guided search through visible inductiveness checks, the algorithm sometimes
gets ``stuck'' in holes of negative counterexamples. While the algorithm does
seem to emerge from these holes eventually, it takes time.

\paragraph*{Comparison to \OneShot.} The \OneShot\ mode uses ``one shot
learning'' rather than an interative CEGIS algorithm. The \OneShot\ algorithm
runs the specification over the smallest 30 elements of the concrete
implementation type, tagging each element as either positive or negative. Doing
so generates sets \Vpos\ and \Vneg, which may be supplied to the synthesizer.
Whatever invariant synthesized is returned as the result. (This algorithm only
works when the specification quantifies over a single element of the abstract
type, which is true for all but 7 of our benchmarks.) Running the \OneShot\
algorithm fails on all but one of our benchmarks, \CF{coq/unique-list-set}, and
does so for a variety of reasons. On some benchmarks, \Myth\ times out,
indicating that the given synthesis problem was too hard, and \Myth\ needed to be
provided fewer examples to find the right invariant. On some
benchmarks, \Myth\ returns the wrong invariant, indicating that the synthesis problem
was underspecified, and too few examples
were given. Merely choosing some fixed number of
examples to build the invariant with is insufficient, that fixed number is too
high for some benchmarks, and too low for others.


\section{Related Work}%
\label{sec:related-work}

\paragraph*{Inferring Representation Invariants.}
To our knowledge, the only prior work that attempts to
automatically infer representation invariants for data structures
is the \Deryaft\ system
by \citet{tacas07/malik/deryaft}, which targets Java classes.
There are three key differences between systems.
First, \Deryaft\ requires a fixed set of predicates (e.g., sortedness)
as an input; the invariants generated are conjunctions of these predicates.
In contrast, \RepInvGen\ can learn new predicates from a general grammar of programs.
Second, the conjunction of predicates that \Deryaft\ produces consists of those predicates
that hold on a fixed set of examples (generated from test executions).
There is no guarantee the final invariant is inductive.
In contrast, \RepInvGen\ employs a CEGIS-based approach to refining
the candidate invariant, terminating only
when a sufficient representation invariant has been identified.
Third, the \Deryaft\ algorithm comes with no theoretical
completeness guarantee.





\paragraph*{Solving constrained Horn clauses.}
Recently, several tools have been developed to infer predicates that
satisfy a given set of {\em constrained Horn clauses} (CHCs).
Inferring representation invariants can be seen as a special case of
CHC solving, since all of our inductiveness constraints are Horn
clauses (e.g., $\Apply{\inv_\star}{s_1} \wedge \Apply{\inv_\star}{s_2}
\implies \Apply{\inv_\star}{(\CF{union}\ s_1\ s_2)})$.  CHCs can
include multiple unknown predicates in their inference problem, whereas there
is only one in ours.

However, existing CHC solvers do not support inference of recursive predicates, which is necessary to handle representation invariants over recursive data types.  Several solvers support only arithmetic constraints~\cite{fmcad17/fedyukovich/freqhorn,
  pldi18/zhu/data-driven-chc, oopsla18/ezudheen/hornice}, while others
support arrays or bit
vectors~\citep{tacas18/champion/hoice,fmsd16/komuravelli/spacer,fmcad18/hojjat/eldarica}
as well.
To our knowledge, 
\Eldarica~\citep{fmcad18/hojjat/eldarica}
is the only CHC solver that supports algebraic data types.
However, \Eldarica\ only computes recursion-free solutions~\citep[Sec III\,C]{fmcad18/hojjat/eldarica}
and therefore cannot express the \emph{sortedness} or
\emph{no-duplicates} properties, for instance.

Though they only handle arithmetic constraints, two of these solvers employ a data-driven technique that is similar to our approach, iterating between a synthesizer and a verifier~\cite{pldi18/zhu/data-driven-chc,oopsla18/ezudheen/hornice}. Like our work, the approach of \citet{pldi18/zhu/data-driven-chc} leverages the observation that handling inductiveness counterexamples
$\Tuple{x}{y}$ is easy when we know that $x$ is constructible.  However, their approach simply checks if a counterexample to full inductiveness happens to have this property, while we exhaustively iterate through these counterexamples until a candidate invariant is \emph{visibly inductive}.  
Intuitively, our approach minimizes the number of inductiveness counterexamples that must be treated heuristically.  
We show that this difference results in a significant performance improvement, and we have proven a completeness result for finite domains, while the approach of \citet{pldi18/zhu/data-driven-chc} lacks a completeness result.  The approach of \citet{oopsla18/ezudheen/hornice} does have a completeness result, and it applies to the infinite domain of integers.  However, they
achieve this guarantee through the use of a specialized synthesizer
designed to handle inductiveness counterexamples directly, while our
approach can use any off-the-shelf synthesizer.  There is also no obvious
analogue to our analysis of higher-order programs in this context.

\paragraph*{Inferring inductive invariants.}  There have been many
techniques developed to infer individual inductive invariants for program verification, for example an inductive invariant for a loop or for a system transition relation.
As discussed in
\cref{sec:a-motivating-example.handling-binary-functions},
module functions may consume or produce multiple arguments or results
of the abstract type, which results in a more general class of inductiveness
counterexamples, whereas loops and transition relations, when viewed as functions, consume
and produce exactly one ``state'' (the analogue of an abstract
value in our setting).

\RepInvGen{} is similar in structure to several data-driven
invariant synthesis engines~\citep{fmsd16/sharma/c2i, pldi16/padhi/loopinvgen, oopsla18/ezudheen/hornice, fse17/nguyen/numinv,
tacas18/neider/incomplete, fmcad19/barbosa/unifpi, nips18/si/code2inv, pldi19/le/sling,
popl16/garg/ice-dt, fmcad17/fedyukovich/freqhorn}.  We experimentally compared our
algorithm to our implementations of the most closely related ones, adapted to
the context of representation invariant synthesis
(\cref{sec:comparisons}).
Broadly, the technical distinctions are
similar to those described above for data-driven CHC solvers.  In particular: (1)
our development of visible inductiveness is novel; (2) aside from
one tool~\cite{popl16/garg/ice-dt} that depends upon a special-purpose
synthesizer to handles inductiveness counterexamples directly,
none are proven complete; (3) they
cannot process higher-order programs; and (4) to our knowledge, none infer recursive invariants. 

The \lstinline{ic3} algorithm for SAT-based model checking employs a notion of {\em relative inductiveness}~\cite{Bradley11}, which is closely related to our notions of conditional and visible inductiveness.  Formally, relative inductiveness is the special case of our conditional inductiveness relation $\LeadsTo{P}{Q}{v}{\tau}{\Valid}$ where $P$ has the form $Q' \wedge Q$ and $\tau = \alpha \rightarrow \alpha$.  The \lstinline{ic3} algorithm uses relative inductiveness to incrementally produce an inductive invariant, by iteratively identifying a state $s$ that leads to a property violation and then conjoining  an inductive strengthening of $\neg s$ to the candidate invariant.  Our notion of visible inductiveness is also used to incrementally produce an inductive invariant, but it works in the opposite direction:   we iteratively identify constructible values of the abstract type and use them to weaken the candidate invariant.  This approach is a natural fit for our data-driven setting.

\paragraph*{Automatic data structure verification.}
The \Leon\ framework \cite{sas11/suter/leon}
can automatically verify correctness of data structure
implementations, but to do so, a user must manually define an
\emph{abstraction function}, which plays a similar role to a representation invariant.
Namely, the abstraction function is a partial function mapping an
element of the concrete type to an element of the abstract type.

There are many techniques for proving properties of heap-based data
structures,
including shape analysis~\cite{Sagiv:1999} and liquid
types~\cite{Kawaguchi:2009}.  These techniques can prove and/or infer
sophisticated invariants, often of imperative code.
However, they are designed to tackle a different problem and do not
infer the inductive representation invariants that are needed to prove modules correct.



\section{Conclusion}%
\label{sec:conclusion}
We present a novel algorithm for synthesis of representation invariants. Our key
insight is that it is possible to drive progress of the algorithm towards its
goal not by eagerly searching for fully inductive invariants, but rather by
searching first for \emph{visibly inductive} invariants. We have proven that our
algorthm is sound and complete, given a sound and complete verifier and
synthesizer, for a first-order type theory with finite types. We also explain
how to extend the algorithm to modules containing higher-order functions, which
involves using contracts to validate and collect objects that cross the module
boundary. We evaluate our algorithm on \NumBenchmarks\ benchmarks and find that
we are able to synthesize \NumSuccess\ of the invariants within 30 minutes (and
most of those in under a minute). Our algorithm is defined independently of the
black-box verifier and synthesizers; as research in verifier and synthesizer
technologies improve, so too will the capabilities of our overall system. While
the tool is currently fully automated, we view this as a step towards an
interactive approach to helping users of proof assistants produce correct
representation invariants.

\begin{acks}
Thanks to our shepherd James Bornholt and the other reviewers for their helpful comments.  This work is supported in part by NSF grant CCF-1837129 and a PhD fellowship from Microsoft Research.
\end{acks}

\bibliography{paper}

\clearpage
\ifappendix
\appendix
\section{Full Example List}
\begin{figure*}%
  \centering%
  \begin{tabulary}{\textwidth}{R|C|C|C|C|C|C|C|C}
    \Name & \Size & \Time\ (s) & \TVT\ (s) & \TVC & \MVT\ (s) & \TST\ (s) & \TSC & \MST\ (s)\\
    \hline \hline \csvreader[head to column names]{generated-data-annotated.csv}{}
    {\Test & \InvariantSize & \MythTime & \MythTotalVerifTime & \MythVerifCalls & \MeanVerifTime &
      \MythTotalSynthTime & \MythSynthCalls & \MeanSynthTime\\}
  \end{tabulary}
  \caption{Information from running \Hanoi\ on our benchmark suite. \Name\ is
    the name of the benchmark. \Size\ is the size of the inferred invariant.
    \Time\ is the time to run the benchmark from start to end. \TVT\ is the
    total time spent verifying. \TVC\ is the total number of verification calls.
    \MVT\ is the average time for a single verification call. \TST\ is the total
    time spent synthesizing. \TSC\ is the total number of synthesis calls. \MST\
    is the average time for a single synthesis call. Benchmarks marked with a *
    were provided an additional function to enable synthesis by \Myth. Elided
    benchmarks are marked in bold.}%
  \label{fig:running-times-two}
\end{figure*}

See Figure~\ref{fig:running-times-two}.

\section{Proofs of Theorems}%
\label{sec:proofs}

\newcounter{origsection}
\newcounter{origtheorem}

\newcommand{\forcenumber}[2]{%
  \setcounter{origsection}{\value{section}}%
  \setcounter{origtheorem}{\value{theorem}}
  \def\thesection{\arabic{section}}%
  \setcounter{section}{#1}%
  \setcounter{theorem}{#2}%
  \addtocounter{theorem}{-1}}

\newcommand{\resetnumber}{%
  \def\thesection{\Alph{section}}%
  \setcounter{section}{\value{origsection}}%
  \setcounter{theorem}{\value{origtheorem}}}

\setlist[itemize]{topsep=0.5em,itemsep=0.375em,partopsep=0pt,parsep=2pt}

\subsection{Additional Definitions}

For convenience, we lift the notion of constructibility to sets and predicates:

\begin{definition}[$\tau$-Constructible Set: {\Constructible[\module]{S}{\tau}}]%
  \label{def:constructible-set}
  A set $S$ of values is said to be $\tau$-constructible using \module,
  denoted \Constructible[\module]{S}{\tau},
  iff every element of $S$ is $\tau$-constructible using \module,
  i.e. $\Forall{v \in S}{\Constructible[\module]{v}{\tau}}$.
\end{definition}

\begin{definition}[$\tau$-Constructible Predicate: {\Constructible[\module]{p}{\tau}}]%
  \label{def:constructible-predicate}
  A predicate $\HasType{p}{\TFun{\Subst{\tau}{\alpha}{\ctyp}}{\TBool}}$
  is said to be $\tau$-constructible using \module,
  denoted \Constructible[\module]{p}{\tau},
  iff the set of values that satisfy $p$ is $\tau$-constructible using \module,
  i.e. $\Constructible[\module]{\{ v \mid p\ v \}}{\tau}$.
\end{definition}

\noindent
We omit \module\ from the subscript when it is clear from context.

\subsection{Type Safety of Conditional Inductiveness Rules}

\begin{lemma}
  If\, \LeadsTo{P}{Q}{v}{\tau}{\Valid} then\, \HasTypeInCtx{v}{\Subst{\tau}{\alpha}{\ctyp}}.
\end{lemma}

\begin{proof}
  We prove this using induction on the derivation of $\LeadsTo{P}{Q}{v}{\tau}{\Valid}$.
  Consider the last rule applied.
  \begin{itemize}
    \item \ref{rule:i-base-valid}: \\
          Trivial since we have that $v = w$, $\tau = \beta$ and \HasTypeInCtx{w}{\beta}.
    \item \ref{rule:i-abs-valid}: \\
          Trivial since we have that $\tau = \alpha$ and \HasTypeInCtx{v}{\ctyp}.
    \item \ref{rule:i-prod-valid}: \\
          We have that
          \begin{enumerate}[leftmargin=1cm,label={\ttfamily(\alph*)}]
            \item $v = \Tuple{v_1}{v_2}$,
            \item $\tau = \TTuple{\tau_1}{\tau_2}$
            \item \LeadsTo{P}{Q}{v_1}{\tau_1}{\Valid}, and
            \item \LeadsTo{P}{Q}{v_2}{\tau_2}{\Valid}.
          \end{enumerate}
          Then, \HasTypeInCtx{v_1}{\Subst{\tau_1}{\alpha}{\ctyp}}
          and \HasTypeInCtx{v_2}{\Subst{\tau_2}{\alpha}{\ctyp}} by induction.
          Hence,
          $$\HasTypeInCtx{\Tuple{v_1}{v_2}}{\Subst{\TTuple{\tau_1}{\tau_2}}{\alpha}{\ctyp}}$$
          and the claim follows.
    \item \ref{rule:i-fun-valid}: \\
          Trivial since we have that
          \HasTypeInCtx{v}{\Subst{\TFun{\sigma_1}{\tau_2}}{\alpha}{\ctyp}}
          and $\tau = \TFun{\sigma_1}{\tau_2}$.
  \end{itemize}\vspace*{-1.6em}
\end{proof}

\begin{lemma}
  If \LeadsToT{P}{Q}{v}{\tau}{S}{V} then \HasTypeInCtx{v}{\Subst{\tau}{\alpha}{\ctyp}}.
\end{lemma}

\begin{proof}
  We prove this using induction on the derivation of $\LeadsToT{P}{Q}{v}{\tau}{S}{V}$.
  Consider the last rule applied.
  \begin{itemize}
    \item \ref{rule:i-abs-cex}: \\
          Trivial since we have that $\tau = \alpha$ and $\{ \}\vdash \HasType{v}{\ctyp}$.
    \item \ref{rule:i-prod-cex-1}: \\
          We have that
          \begin{enumerate}[leftmargin=1cm,label={\ttfamily(\alph*)}]
            \item $v = \Tuple{v_1}{v_2}$,
            \item $\tau = \TTuple{\tau_1}{\tau_2}$
            \item \LeadsToT{P}{Q}{v_1}{\tau_1}{S}{V}, and
            \item \HasTypeInCtx{v_2}{\Subst{\tau_2}{\alpha}{\ctyp}}.
          \end{enumerate}
          Then, \HasTypeInCtx{v_1}{\Subst{\tau_1}{\alpha}{\ctyp}} by induction.
          Hence,
          $$\HasTypeInCtx{\Tuple{v_1}{v_2}}{\Subst{\TTuple{\tau_1}{\tau_2}}{\alpha}{\ctyp}}$$
          and the claim follows.
    \item \ref{rule:i-prod-cex-2}: \\
          Similar to \ref{rule:i-prod-cex-1} above.
    \item \ref{rule:i-fun-cex}: \\
          Trivial since we have that
          \HasTypeInCtx{v}{\Subst{\TFun{\sigma_1}{\tau_2}}{\alpha}{\ctyp}}
          and $\tau = \TFun{\sigma_1}{\tau_2}$.
  \end{itemize}\vspace*{-1.6em}
\end{proof}

\subsection{Soundness of \RepInvGen}

\begin{lemma}\label{lem:NoNegatives-sound}
  If \:\verifier\ is sound, then
  $$
    \Apply{\NoNegatives}{\inv} = \Valid
    \implies
    \Satisfies{\module}{\interface}{\inv}{\spec}
  $$
\end{lemma}

\begin{proof}
  We explicitly check for sufficiency at \cref{line:sufficiency-check},
  and for inductiveness at \cref{line:inductiveness-check}.
  \NoNegatives\ returns \Valid\ only if both these checks succeed.
  Since \verifier\ is assumed to be sound, the claim immediately follows.
\end{proof}

\forcenumber{3}{9}
\begin{theorem}\label{apx.thm:RepInvGen-sound}
  If \:\verifier\ is sound, then \RepInvGen\ is sound.
\end{theorem}
\resetnumber

\begin{proof}
  Since \RepInvGen\ returns an invariant \inv\ only at \cref{line:repinvgen-result},
  when $\Apply{\NoNegatives}{\inv} = \Valid$,
  the claim immediately follows from \cref{lem:NoNegatives-sound}.
\end{proof}

\subsection{Completeness of \RepInvGen}

\begin{lemma}\label{lem:triangle-reverse-Constructible}
  $ \Constructible{Q}{\alpha} \wedge \LeadsTo{P}{Q}{v}{\sigma}{\Valid}
    \implies
    \Constructible{v}{\sigma} $.
\end{lemma}

\begin{proof}
  We prove this using induction on the derivation of $\LeadsTo{P}{Q}{v}{\sigma}{\Valid}$.
  Consider the last rule applied.
  \begin{itemize}
    \item \ref{rule:i-base-valid}: \\
          Trivial since values of type $\beta$ are always Constructible.
    \item \ref{rule:i-abs-valid}: \\
          Then $\sigma = \alpha$ and $\Apply{Q}{v}$.
          Then since \Constructible{Q}{\alpha},
          the claim follows immediately from \cref{def:constructible-predicate}.
    \item \ref{rule:i-prod-valid}: \\
          We have that
          \begin{enumerate}[leftmargin=1cm,label={\ttfamily(\alph*)}]
            \item $\sigma = \TTuple{\sigma_1}{\sigma_2}$,
            \item $v = \Tuple{v_1}{v_2}$,
            \item $\LeadsTo{P}{Q}{v_1}{\sigma_1}{\Valid}$, and
            \item $\LeadsTo{P}{Q}{v_2}{\sigma_2}{\Valid}$.
          \end{enumerate}
          Then by induction we have that \Constructible{v_1}{\sigma_1} and \Constructible{v_2}{\sigma_2},
          hence also \Constructible{v}{\sigma}.
  \end{itemize}\vspace*{-1.6em}
\end{proof}

\begin{lemma}\label{lem:triangle-collectv}
  $ \LeadsTo{P}{Q}{v}{\sigma}{\Valid}
    \ \implies\
    \Forall{x \in \CollectV{\sigma}{v}}{\Apply{Q}{x}} $.
\end{lemma}

\begin{proof}
  We prove this using induction on the derivation of $\LeadsTo{P}{Q}{v}{\sigma}{\Valid}$.
  Consider the last rule applied.
  \begin{itemize}
    \item \ref{rule:i-base-valid}: \\
          Trivial since $\CollectV{\beta}{\W} = \emptyset$.
    \item \ref{rule:i-abs-valid}: \\
          We have $\HasTypeInCtx{v}{\ctyp}$ and $\Apply{Q}{v}$.
          Since $\sigma = \alpha$, due to \ref{rule:c-abs}
          we also have $\CollectV{\alpha}{v} = \{v\}$.
          Then the claim follows.
    \item \ref{rule:i-prod-valid}: \\
          We have that
          \begin{enumerate}[leftmargin=1cm,label={\ttfamily(\alph*)}]
            \item $\sigma = \TTuple{\sigma_1}{\sigma_2}$,
            \item $v = \Tuple{v_1}{v_2}$,
            \item $\LeadsTo{P}{Q}{v_1}{\sigma_1}{\Valid}$, and
            \item $\LeadsTo{P}{Q}{v_2}{\sigma_2}{\Valid}$.
          \end{enumerate}
          Since $\CollectV{\TTuple{\sigma_1}{\sigma_2}}{\Tuple{v_1}{v_2}} = \CollectV{\sigma_1}{v_1} \cup \CollectV{\sigma_2}{v_2}$
          by \ref{rule:c-prod}, the claim follows by induction.
  \end{itemize}\vspace*{-1.6em}
\end{proof}

\begin{lemma}\label{lem:collectv-triangle}
  $ \Forall{x \in \CollectV{\sigma}{v}}{\Apply{Q}{x}}
    \implies
    \LeadsTo{P}{Q}{v}{\sigma}{\Valid} $.
\end{lemma}

\begin{proof}
  We prove this using induction on the derivation of $\CollectV{\sigma}{v}$.
  Consider the last rule applied:
  \begin{itemize}
    \item \ref{rule:c-base}: \\
          Trivial since $\CollectV{\beta}{\W} = \emptyset$.
    \item \ref{rule:c-abs}: \\
          We have that
          \begin{enumerate}[leftmargin=1cm,label={\ttfamily(\alph*)}]
            \item $\sigma = \alpha$,
            \item $\HasTypeInCtx{v}{\ctyp}$, and
            \item $\CollectV{\alpha}{v} = \{v\}$.
          \end{enumerate}
          We also have $\Apply{Q}{v}$ by assumption.
          Then, we have that $\LeadsTo{P}{Q}{v}{\alpha}{\Valid}$ by \ref{rule:i-abs-valid}.
    \item \ref{rule:c-prod}: \\
          We have that $\CollectV{\TTuple{\sigma_1}{\sigma_2}}{\Tuple{v_1}{v_2}} = \CollectV{\sigma_1}{v_1} \cup \CollectV{\sigma_2}{v_2}$.
          Also, since $\Forall{x \in \CollectV{\sigma_1}{v_1} \cup \CollectV{\sigma_2}{v_2}}{\Apply{Q}{x}}$ by assumption,
          we have that $\LeadsTo{P}{Q}{v_1}{\sigma_1}{\Valid}$ and $\LeadsTo{P}{Q}{v_2}{\sigma_2}{\Valid}$ by induction.
          Then, $\LeadsTo{P}{Q}{\Tuple{v_1}{v_2}}{\TTuple{\sigma_1}{\sigma_2}}{\Valid}$ follows by \ref{rule:i-prod-valid}.
  \end{itemize}\vspace*{-1.6em}
\end{proof}

\begin{lemma}\label{lem:triangle-counterexample}
  If $\LeadsToT{P}{Q}{v}{\tau}{S}{V}$, then
  \begin{inlist}
    \item $\Forall{x \in S}{\Apply{P}{x}}$,
    \item $V \neq \emptyset$, and
    \item $\Forall{x \in V}{\neg\Apply{Q}{x}}$.
  \end{inlist}
\end{lemma}

\begin{proof}
  We prove this using induction on the derivation of $\LeadsToT{P}{Q}{v}{\tau}{S}{V}$.
  Consider the last rule applied.
  \begin{itemize}
    \item \ref{rule:i-abs-cex}: \\
          Straightforward since $S = \emptyset$, $V = \{ v \}$.
    \item \ref{rule:i-prod-cex-1} and \ref{rule:i-prod-cex-2}: \\
          Follow immediately from the induction hypothesis.
    \item \ref{rule:i-fun-cex}: \\
          We have that
          \begin{enumerate}[leftmargin=1cm,label={\ttfamily(\alph*)}]
            \item $\LeadsTo{Q}{P}{v_1}{\sigma_1}{\Valid}$, and
            \item $\LeadsToT{P}{Q}{v_2}{\tau_2}{S'}{V}$.
          \end{enumerate}
          Then, (2) and (3) are satisfied by the induction hypothesis.
          Due to \cref{lem:triangle-collectv}, we have that $\Forall{x \in \CollectV{\sigma_1}{v_1}}{\Apply{P}{x}}$.
          Then, since $\Forall{x \in S'}{\Apply{P}{x}}$ from the induction hypothesis,
          we have that $\Forall{x \in \CollectV{\sigma_1}{v_1} \cup S'}{\Apply{P}{x}}$.
          Hence, (1) is satisfied as well since $S = \CollectV{\sigma_1}{v_1} \cup S'$.
  \end{itemize}\vspace*{-1.6em}
\end{proof}

\begin{lemma}\label{lem:triangle-constructible}
  If $\Constructible{v}{\tau} \wedge \Constructible{P}{\alpha} \wedge \LeadsToT{P}{Q}{v}{\tau}{S}{V}$
  then $\Constructible{V}{\alpha}$.
\end{lemma}

\begin{proof}
  We prove this using induction on the derivation of $\LeadsToT{P}{Q}{v}{\tau}{S}{V}$.
  Consider the last rule applied.
  \begin{itemize}
    \item \ref{rule:i-abs-cex}: \\
          Straightforward since we have $\alpha = \tau$ and $\Constructible[\module]{v}{\tau}$ holds by assumption.
    \item \ref{rule:i-prod-cex-1}: \\
          We have $\LeadsToT{P}{Q}{v_1}{\tau_1}{S}{V}$,
          and it is easy to show that $\Constructible{\Tuple{v_1}{v_2}}{\TTuple{\tau_1}{\tau_2}} \implies \Constructible{v_1}{\tau_1}$.
          Then the claim follows by induction.
    \item \ref{rule:i-prod-cex-2}: \\
          Similar to \ref{rule:i-prod-cex-1} above.
    \item \ref{rule:i-fun-cex}: \\
          We have that
          \begin{enumerate}[leftmargin=1cm,label={\ttfamily(\alph*)}]
            \item $\LeadsTo{Q}{P}{v_1}{\sigma_1}{\Valid}$, and
            \item $\LeadsToT{P}{Q}{v_2}{\tau_2}{S'}{V}$.
          \end{enumerate}
          Due to \cref{lem:triangle-reverse-Constructible}, we also have that $\Constructible{v_1}{\sigma_1}$.
          Hence $\Constructible{v_2}{\tau_2}$ since $\EvaluatesTo{\Apply{v}{v_1}}{v_2}$,
          and then the claim follows by induction.
  \end{itemize}\vspace*{-1.6em}
\end{proof}

\begin{lemma}\label{lem:triangle-valid}
  If $\LeadsTo{P}{Q}{v}{\tau}{\Valid} \wedge \LeadsToT{Q}{Q}{v}{\tau}{S}{V}$, then
  $\Exists{x \in S}{\neg\Apply{P}{x}}$.
\end{lemma}

\begin{proof}
  We prove this using induction on the derivation of $\LeadsToT{Q}{Q}{v}{\tau}{S}{V}$.
  Consider the last rule applied.
  \begin{itemize}
    \item \ref{rule:i-abs-cex}: \\
          Impossible, since $\LeadsTo{P}{Q}{v}{\tau}{\Valid} \implies \Apply{Q}{v}$.
    \item \ref{rule:i-prod-cex-1}: \\
          We have $\LeadsToT{Q}{Q}{v_1}{\tau_1}{S}{V}$.
          Due to the assumption $\LeadsTo{P}{Q}{\Tuple{v_1}{v_2}}{\TTuple{\tau_1}{\tau_2}}{\Valid}$,
          and \ref{rule:i-prod-valid}, we also have $\LeadsTo{P}{Q}{v_1}{\tau_1}{\Valid}$.
          Then, the claim follows immediately from the induction hypothesis.
    \item \ref{rule:i-prod-cex-2}: \\
          Similar to \ref{rule:i-prod-cex-1} above.
    \item \ref{rule:i-fun-cex}: \\
          We have that
          \begin{enumerate}[leftmargin=1cm,label={\ttfamily(\alph*)}]
            \item $\LeadsTo{Q}{Q}{v_1}{\sigma_1}{\Valid}$, and
            \item $\LeadsToT{Q}{Q}{v_2}{\tau_2}{S'}{V}$.
          \end{enumerate}
          Since $\LeadsTo{P}{Q}{v}{\tau}{\Valid}$, we have two possibilities:
          \begin{itemize}
            \item $\LeadsTo{Q}{P}{v_1}{\sigma_1}{\Valid}$: \\
                  We have $\LeadsTo{P}{Q}{v_2}{\tau_2}{\Valid}$.
                  Then, due to \texttt{(b)} above, we have $\Exists{x \in S'}{\neg\Apply{P}{x}}$ by induction.
            \item $\neg ( \LeadsTo{Q}{P}{v_1}{\sigma_1}{\Valid} )$: \\
                  Due to (the contrapositive of) \cref{lem:collectv-triangle},
                  we have that $\Exists{x \in \CollectV{\sigma_1}{v_1}}{\neg \Apply{P}{x}}$.
          \end{itemize}
          The claim follows since $S = S' \cup \CollectV{\sigma_1}{v_1}$.
  \end{itemize}\vspace*{-1.6em}
\end{proof}

\begin{lemma}\label{lem:ClosedPositives-sound-and-complete}
  If \:\verifier\ is sound and complete, and \Vpos\ and \inv\ satisfy
  $ \Constructible{\Vpos}{\alpha} \wedge \Forall{v \in \Vpos}{\Apply{\inv}{v}} $, then
  \begin{enumerate}
    \item $\Apply{\ClosedPositives}{\Vpos}{\inv} = \CounterExample{V} $ \\
          $ \implies V \neq \emptyset \ \wedge\ \Constructible{V}{\alpha} \ \wedge\ \left(\Forall{v \in V}{\neg\Apply{\inv}{v} \wedge v \not\in \Vpos} \right) $
    \item $\Apply{\ClosedPositives}{\Vpos}{\inv} = \Valid $ \\
          $ \implies \LeadsTo{\Vpos}{\inv}{\mval}{\mtyp}{\Valid} $
  \end{enumerate}
\end{lemma}

\begin{proof}
  On expanding \IsCondInductive{}, (2) immediately follows.
  For (1), we first observe that \mval\ is trivially constructible,
  and thus $\Constructible{V}{\alpha}$ by \cref{lem:triangle-constructible}.
  Due to \cref{lem:triangle-counterexample},
  we also have that $V \neq \emptyset \wedge \Forall{x \in V}{\neg \Apply{\inv}{x}}$.
  Finally, we have that $\Forall{x \in V}{v \not\in \Vpos}$
  since $\Forall{v \in \Vpos}{\Apply{\inv}{v}}$ by assumption.
\end{proof}

\begin{lemma}\label{lem:NoNegatives-sound-and-complete}
  If \:\verifier\ is sound and complete, then
  \begin{enumerate}
    \item $\Apply{\NoNegatives}{\inv} = \CounterExample{S} $ \\
          $ \implies \left( S \neq \emptyset \ \ \wedge \ \ \Forall{v \in S}{\Apply{\inv}{v} \wedge \neg\Apply{\spec}{v}} \right) $ \\
          \hspace*{4pt}$ \vee\mkern12mu \left( \Forall{v \in S}{\Apply{\inv}{v}} \ \ \wedge \ \ \Exists{V}{\LeadsToT{\inv}{\inv}{\mval}{\mtyp}{S}{V}} \right)$
    \item $\Apply{\NoNegatives}{\inv} = \Valid $ \\
          $ \implies \LeadsTo{\inv}{\inv}{\mval}{\mtyp}{\Valid} $
  \end{enumerate}
\end{lemma}

\begin{proof}
  On expanding \IsCondInductive{}, (2) immediately follows.
  (1) follows from \cref{lem:triangle-counterexample} and \cref{def:sufficiency}.
\end{proof}

Finally, we prove a generalization of the soundness and completeness theorem stated in the main body of the paper.

\forcenumber{3}{10}
\begin{theorem}\label{apx.thm:RepInvGen-complete-finite}
  If \:\verifier\ and \:\synthesizer\ are both sound and complete,
  then \RepInvGen\ \Vpos\ \Vneg\ is sound and complete for any finite domain \ctyp{}
  whenever $\Constructible{\Vpos}{\alpha} \ \wedge\ \Vpos \cap \Vneg = \emptyset$.
\end{theorem}
\resetnumber

\begin{proof}[Proof]
  Soundness follows from \cref{apx.thm:RepInvGen-sound}.
  Using a lexicographic ranking argument,
  we now show that \RepInvGen\ generates a predicate in finite steps.

  Consider the ranking function:
  $$ \rank(\Vpos, \inv) \definedas \Tuple{\card{\ctyp} - \card{\Vpos}}{\card{\ctyp} - \card{\Vneg}} $$
  It is lower bounded by $\Tuple{0}{0}$,
  and in the remainder of the proof, we show that it decreases lexicographically
  with each recursive \RepInvGen\ call.

  First, we note that since $\Vpos \cap \Vneg = \emptyset$ and \synthesizer\ is complete,
  a candidate invariant \inv\ will always be obtained.
  Moreover, since \synthesizer\ is sound, we also have
  \begin{equation}\label{eqn:synth-sound-result}
    \Forall{v \in \Vpos}{\Apply{\inv}{v}}
    \quad\wedge\quad
    \Forall{v \in \Vneg}{\neg\Apply{\inv}{v}}
  \end{equation}

  We have two cases for the \ClosedPositives\ call at \cref{line:positive-check}:
  \begin{enumerate}
    \item $\Apply{\ClosedPositives}{\Vpos}{\inv} = \CounterExample{P}$.

          Due to \cref{lem:ClosedPositives-sound-and-complete},
          we have $\Constructible{P}{\alpha} $.
          Since we reset $\Vneg = \emptyset$ for the recursive call at \cref{line:positive-recurse},
          we have that
          $$ \Constructible{\Vpos \cup P}{\alpha} \quad\wedge\quad (\Vpos \cup P) \cap \emptyset = \emptyset $$

          Since $P \neq \emptyset \wedge \Forall{v \in P}{v \not\in \Vpos}$,
          also due to \cref{lem:ClosedPositives-sound-and-complete},
          we have that $\card{\Vpos \cup P} > \card{\Vpos}$.
          Thus, the first index of \rank\ decreases at the recursive call.

    \item $\Apply{\ClosedPositives}{\Vpos}{\inv} = \Valid$.

          We have two cases for the \NoNegatives\ call at \cref{line:negative-check}:
          \begin{enumerate}
            \item $\Apply{\NoNegatives}{\inv} = \CounterExample{N} $.

                  For the recursive call at \cref{line:negative-recurse}, \Vpos\ remains unchanged
                  and clearly $\Vpos \cap (\Vneg \cup (N \setminus \Vpos)) = \emptyset$ since $\Vpos \cap \Vneg = \emptyset$. \\[0.5em]
                  First we note that $\LeadsTo{\Vpos}{\inv}{\mval}{\mtyp}{\Valid}$,
                  due to \cref{lem:ClosedPositives-sound-and-complete}.
                  Applying $\vee$-elimination to case (1) of \cref{lem:NoNegatives-sound-and-complete},
                  we have the following two cases,
                  and we show that $N \neq \emptyset \wedge \Forall{v \in N}{\Apply{\inv}{v}} \wedge (N \setminus \Vpos) \neq \emptyset$ holds for both.
                  \begin{enumerate}
                    \item $N \neq \emptyset \ \ \wedge \ \ \Forall{v \in N}{\Apply{\inv}{v} \wedge \neg\Apply{\spec\,[\ctyp]}{\mval}{v}}$. \\[0.5em]
                          Due to (the contrapositive of) \cref{thm:exists-repinv-implies-constructible},
                          we have that $\Forall{v \in N}{\neg\Constructible{v}{\alpha}}$
                          and so $(N \setminus \Vpos) \neq \emptyset$.

                    \item $\Forall{v \in N}{\Apply{\inv}{v}} \ \ \wedge \ \ \Exists{V}{\LeadsToT{\inv}{\inv}{\mval}{\mtyp}{N}{V}}$.

                          Due to \cref{lem:triangle-valid}, $\Exists{v \in N}{v \not\in \Vpos}$.
                          Thus, we have that $N \neq \emptyset \wedge (N \setminus \Vpos) \neq \emptyset$.

                  \end{enumerate}
                  Due to \cref{eqn:synth-sound-result}, $\Forall{v \in N}{v \not\in \Vneg}$.
                  Thus, in both cases we have that
                  $\card{\Vneg \cup (N \setminus \Vpos)} > \card{\Vneg}$,
                  and the second index of \rank\ decreases at the recursive call.

            \item $\Apply{\NoNegatives}{\inv} = \Valid $.

                  We have a solution.
          \end{enumerate}
  \end{enumerate}\vspace*{-1.6em}
\end{proof}

\clearpage
\fi

\end{document}